%% file: main.tex
\documentclass[acmsmall, screen]{acmart}
\usepackage{amsmath}
\usepackage{amsthm,thm-restate}
\usepackage{graphicx}
\usepackage{wrapfig}
\usepackage{caption}
\usepackage{subcaption}
\usepackage{enumitem}
\usepackage{bbm}
\usepackage{xspace}
\input{header}

\title{SPLIT: SymPathy for Large jobs Improves Tail latency}

\author{Zhouzi Li}
\email{zhouzil@andrew.cmu.edu}
\affiliation{
  \institution{Carnegie Mellon University}
  \department{Computer Science Department}
  \city{Pittsburgh}
  \state{Pennsylvania}
  \country{United States}
}

\author{Mor Harchol-Balter}
\email{harchol@cs.cmu.edu}
\affiliation{
  \institution{Carnegie Mellon University}
  \department{Computer Science Department}
  \city{Pittsburgh}
  \state{Pennsylvania}
  \country{United States}
}

\author{Alan Scheller-Wolf}
\email{awolf@andrew.cmu.edu}
\affiliation{
  \institution{Carnegie Mellon University}
  \department{Tepper School of Business}
  \city{Pittsburgh}
  \state{Pennsylvania}
  \country{United States}
}

\begin{document}

\input{0-abstract}
\maketitle

\input{1-Intro}
\input{2-prior}
\input{3-setting}
\input{4-SPLIT}
\input{5-Thresh-SPLIT}
\input{6-TAG-SPLIT}
\input{7-simulation}
\input{8-discussion}

\input{9-conclusion}

\begin{acks}
This work was supported by NSF-CIF-2403194, NSF-III-2322973, and NSF-CMMI-2307008.
\end{acks}

\bibliographystyle{ACM-Reference-Format}
\bibliography{bib}

\appendix
\input{Appendix}

\end{document}

%% file: header.tex
\newif\ifshowchanges
\showchangesfalse

\makeatletter
\ifshowchanges
  \NewDocumentCommand{\removed}{+m}{\textcolor{red}{#1}}
  \NewDocumentCommand{\added}{+m}{\textcolor{blue}{#1}}
  \NewDocumentEnvironment{removedblock}{}{\color{red}}{}
  \NewDocumentEnvironment{addedblock}{}{\color{blue}}{}
\else
  \NewDocumentCommand{\removed}{+m}{\@bsphack\@esphack}
  \NewDocumentCommand{\added}{+m}{#1}
  \NewDocumentEnvironment{removedblock}{+b}{}{}
  \NewDocumentEnvironment{addedblock}{}{}{}
\fi
\makeatother

\newcommand{\E}[1]{\mathrm{\bf E}\left[#1\right]}
\def\P#1{\mathrm{\mathbb{P}}\left[#1\right]}

\usepackage{xparse}
\usepackage{xspace}

\newcommand{\threshsplit}[1]{\mbox{Thresh-SPLIT}(\ensuremath{#1})\xspace}
\newcommand{\TAGsplit}[1]{\mbox{TAGS-SPLIT}(\ensuremath{#1})\xspace}

\newcommand{\SPLIT}{SPLIT\xspace}
\newcommand{\WA}[2]{W_{A_{#1}}(#2)}
\newcommand{\Wsuper}[1]{W_{#1}^{super}}
\newcommand{\Wstar}[1]{W_{#1}^{SRPT-1}}
\newcommand{\WAstar}[2]{W_{A_{#1}}^{SRPT-1}(#2)}

\newtheorem{approximation}{Approximation}

%% file: 0-abstract.tex
\begin{abstract}
We study the asymptotic response time tail in the M/G/$n$ multi-server queue with heavy-tailed (regularly varying) job sizes, a setting representative of modern computing workloads. For single-server systems, tail optimization is well understood: under heavy-tailed job sizes, policies such as SRPT that strictly prioritize short jobs are strongly tail optimal, and giving any priority to large jobs is harmful. For multi-server systems, the question has been almost entirely open.

This paper gives the first strongly tail-optimal scheduling policies for the M/G/$n$ queue with heavy-tailed job sizes. Our central finding is that the multi-server case is intrinsically different from the single-server case: giving a small amount of ``sympathy'' to large jobs is essential for strong tail optimality. We establish strong (or arbitrarily close to strong) tail optimality across the full stability region, both with and without knowledge of job sizes.
\end{abstract}

%% file: 1-Intro.tex
\section{Introduction}

Optimal scheduling has been an important topic in the performance modeling community for many decades. In the simplest setting, where there is just a single server and the goal is to minimize mean response time, the optimal scheduling policy is well understood: at all times, preemptively serve the job with the Shortest Remaining Processing Time (SRPT) \cite{schrage1968proof}. The intuition is that a single large job occupying the server forces many smaller jobs to wait, inflating the response time of all of them; SRPT avoids this by always preempting in favor of shorter jobs. This intuition is so fundamental that prioritizing short jobs has become almost synonymous with optimal scheduling.

However, in addition to optimizing the mean response time, another important question is that of \emph{optimal tail scheduling}. In many modern computing systems, the tail of the response time distribution matters at least as much as the mean: Metrics such as the 99th or 99.9th percentile of response time, often codified as service-level objectives (SLOs), directly govern user experience. 

This paper develops the first tail optimal scheduling policies for multiserver systems with heavy-tailed job sizes. We formally define tail optimality in Section~\ref{sec:intro:tail}, briefly survey existing tail optimization results in Section~\ref{sec:intro:summary}, and explain why the multi-server case remains open in Section~\ref{sec:intro:open}. We then present the main intuition behind our approach in Section~\ref{sec:intro:intuition}.

\subsection{Definition of tail optimality}
\label{sec:intro:tail}

A common theoretical formalization of tail optimization studies the asymptotic decay of $\P{T > t}$ as $t \to \infty$ \cite{wierman2012tail,boxma2007tails,nuyens2008preventing,yu2024strongly,yu2025tale,harlev2025gittins}. As argued in \cite{yu2024strongly}, there are several reasons to focus on this asymptotic regime: (i) optimizing $\P{T > t}$ for any single value of $t$ is seldom the sole design objective; rather, one hopes to achieve low $\P{T > t}$ across a range of $t$; (ii) practical SLOs relate to high-quantile response times, which correspond to large values of $t$; and (iii) optimizing $\P{T > t}$ for fixed finite $t$ appears to be theoretically intractable, whereas the asymptotic $t \to \infty$ limit has seen promising recent progress.

To make the notion of optimality precise, we say a scheduling policy $\pi$ is \emph{weakly tail optimal} if there exists a constant $c \geq 1$ such that 
\[
  \sup_{\pi'} \limsup_{t\to\infty} \frac{\P{T^\pi>t}}{\P{T^{\pi'}>t}} = c,
\]
and \emph{strongly tail optimal} if additionally $c = 1$. 

\subsection{Brief summary of tail optimization results}
\label{sec:intro:summary}

For {\em single-server} systems, tail optimization is by now well understood, though the answer depends critically on whether the job size distribution is heavy-tailed or light-tailed. As shown by \cite{wierman2012tail}, no single policy can be weakly tail optimal for both heavy-tailed and light-tailed distributions simultaneously; the two settings require fundamentally different scheduling strategies.

Heavy-tailed (regularly varying) job size distributions are particularly important because they are widely observed in real-world computer workloads, including web file sizes \cite{crovella1997self}, Unix process lifetimes \cite{harchol1997exploiting}, and modern cloud and datacenter job sizes \cite{tirmazi2020borg}. Under heavy-tailed job sizes, a number of policies, including SRPT, Processor Sharing (PS), Foreground-Background (FB), and SMART, have been shown to be strongly tail optimal \cite{nuyens2008preventing,borst2006sojourn,guillemin2004tail,boxma2007tails,scully2025does}. Together with the lower bound in 
\cite{wierman2012tail,yu2024strongly}, these results completely characterize single-server tail optimality under heavy tails.
Under light-tailed job sizes, the story is different: SRPT has the worst tail decay rate, while FCFS is weakly but not strongly tail optimal \cite{boxma2007tails,grosof2021nudge}; the strongly optimal policy Boost was proposed recently in \cite{yu2024strongly}. We defer a more detailed discussion of the light-tailed case to Section~\ref{sec:prior:tail-opt}.

Table~\ref{tab:single-server} summarizes the state of affairs for single-server systems.

\begin{table}[ht]
\centering
\captionsetup{skip=4pt}
\caption{ Response time tail optimality results for single-server (M/G/1) systems.}
\label{tab:single-server}
\begin{tabular}{l|l|l}
\hline
 & \textbf{Light-tailed Sizes} & \textbf{Heavy-tailed Sizes} \\
\hline
\textbf{Weakly optimal} & FCFS \cite{boxma2007tails}, Nudge \cite{grosof2021nudge} &  \\
\hline
\textbf{Strongly optimal (known sizes)} & Boost \cite{yu2024strongly} & SRPT \cite{nuyens2008preventing}, SMART \cite{nuyens2008preventing} \\
\hline
\textbf{Strongly optimal (unknown sizes)} & $\gamma$-Gittins \cite{harlev2025gittins} & PS \cite{borst2006sojourn,guillemin2004tail}, FB \cite{nuyens2008preventing}, Gittins \cite{scully2025does} \\
\hline
\end{tabular}
\end{table}

\subsection{Tail optimality is open for multi-server systems}
\label{sec:intro:open}

The {\em multi-server} setting, by contrast, remains almost entirely open. The only prior result on tail-optimal scheduling in the M/G/$n$ is due to \cite{yu2025tale}, which studies light-tailed job sizes and shows that Boost is strongly tail optimal in the heavy-traffic regime; however, the performance of Boost degrades substantially outside of heavy traffic. We discuss this work in detail in Section~\ref{sec:prior:tail-opt}. To the best of our knowledge, no existing policy has been proved to be strongly -- or even weakly -- tail optimal for multi-server systems with heavy-tailed job sizes, which is the setting most representative of real-world systems and workloads. This paper is the first to fill this gap. Table~\ref{tab:multi-server} summarizes the known results for multi-server systems, including the contributions of this paper.

\begin{table}[ht]
\centering
\captionsetup{skip=4pt}
\captionbox{Tail optimality results for multi-server (M/G/$n$) systems.\newline
$^*$In the case when system load $\rho<\frac{n-1}{n}$.\newline
$^\dagger$Arbitrarily close to strongly optimal.\newline
$^+$Based on Approximation \ref{approx:TAGS}, as used in \cite{harchol2002task, BDS2019, BDS2020}.\label{tab:multi-server}}{%
\begin{tabular}{l|l|l}
\hline
 & \textbf{Light-tailed Sizes} & \textbf{Heavy-tailed Sizes} \\
\hline
\textbf{Weakly optimal} &  & SRPT-$n$ [\textbf{This paper}] \\
\hline
\textbf{Strongly optimal} & Boost \cite{yu2025tale} (heavy traffic only) & \SPLIT$^*$ [\textbf{This paper}] \\
\textbf{(known sizes)} &  & $\threshsplit{d}^\dagger$ [\textbf{This paper}] \\
\hline
\textbf{Strongly optimal} &  & $\TAGsplit{d}^{\dagger +}$ [\textbf{This paper}] \\
\textbf{(unknown sizes)} &  &  \\
\hline
\end{tabular}}
\end{table}

We study tail-optimal scheduling in the M/G/$n$ with heavy-tailed (regularly varying) job sizes. Our central finding is that optimal scheduling for multi-server systems is \emph{intrinsically different} from optimal scheduling for single-server systems. On the one hand, we prove that SRPT-$n$ (where the $n$ servers at all times serve the $n$ jobs with the shortest remaining processing time) is still \emph{weakly} tail optimal (Theorem~\ref{theorem: SRPT-n is weakly tail optimal new})\added{, but it is not \emph{strongly} tail optimal (Theorem~\ref{thm: SRPT-n not strongly tail optimal})}. On the other hand, we find that achieving \emph{strong} tail optimality requires giving some ``sympathy'' to large jobs, a strategy that would be harmful in a single-server system.  

\subsection{Intuition behind our approach \SPLIT}
\label{sec:intro:intuition}

To understand why multi-server systems require a different approach, it is helpful to contrast the single-server and multi-server cases.

In a {\em single-server} system, the main concern is preventing small jobs from being stuck behind large ones. Under heavy-tailed job sizes, being blocked by a single large job can be catastrophic for the many small jobs waiting behind it. One might consider giving large jobs some help to improve the response time tail, since it is the large jobs that dominate the tail. However, any service given to a large job comes directly at the expense of small jobs waiting behind it: There is no way to help large jobs without hurting small ones. This is roughly why SRPT, which strictly prioritizes short jobs, is strongly tail optimal under heavy tails \cite{nuyens2008preventing}. Even under light-tailed job sizes, where the optimal strategies are completely different and FCFS is weakly optimal, policies like Nudge \cite{grosof2021nudge} and Boost \cite{yu2024strongly} still rely on giving short jobs extra priority over large ones to improve upon FCFS.

In a {\em multi-server} system, however, the situation is fundamentally different. Even if a large job occupies one server, the remaining servers can continue to process small jobs without interruption.  This means we may be able to afford to give large jobs some sympathy, dedicating service capacity to them, without significantly hurting small jobs in the process. Since it is the large jobs that dominate the response time tail, this sympathy could improve tail performance. Moreover, an added perk of giving large jobs some priority is improved \emph{packing}: Under SRPT-$n$, at the end of a busy period, only a few large jobs remain, resulting in underutilization of the servers (see \cite{grosof2025outperforming} for detailed discussion). This underutilization can potentially be mitigated by finishing off some of these larger jobs early on. We discuss the effect of packing in detail in Section~\ref{sec:discussion:packing}.

These ideas motivate our new policy \SPLIT (Sympathy for Large jobs Improves Tail latency). Consider first the case where system load $\rho < \frac{n-1}{n}$. Here the system can afford to sacrifice one entire server to serve a large job, while the remaining servers continue to drain small jobs. In this regime, we design the \SPLIT policy: The policy splits the $n$ servers into two groups, where $n-1$ servers run SRPT-$(n{-}1)$, ensuring that small jobs are processed efficiently, while the remaining one server runs a non-preemptive Largest-Job-First (LJF) policy, giving large jobs dedicated service. We prove that \SPLIT is strongly tail optimal for any load $\rho<\frac{n-1}{n}$ (Section~\ref{sec:SPLIT}).

When $\rho \geq \frac{n-1}{n}$, dedicating an entire server to large jobs is no longer feasible: Imagine a very large job occupying a server. During this long period of time, all arrivals must be served by the remaining $n-1$ servers, and the system becomes temporarily unstable, which inflates the response time of all small jobs arriving during this period of time. 
 To handle this regime, we design $\threshsplit{d}$, which uses a size threshold $d$ to control the division of work. Jobs smaller than $d$ are routed to the $n-1$ servers, which are guaranteed (by the choice of $d$) a load strictly below $1$, while jobs larger than $d$ are handled by one server running SRPT-1. By choosing $d$ appropriately, $\threshsplit{d}$ can be made arbitrarily close to strongly tail optimal under any load $\rho < 1$ (Section~\ref{sec:threshsplit}). 

Finally, in Section~\ref{sec:TAG-SPLIT}, we consider the setting where exact job sizes are unknown. Motivated by the Task Assignment by Guessing Size (TAGS) policy in \cite{harchol2002task}, we generalize the $\threshsplit{d}$ policy to the $\TAGsplit{d}$ policy. 
We show that by choosing the threshold $d$ appropriately, $\TAGsplit{d}$ can also be made arbitrarily close to strongly tail optimal under all loads $\rho < 1$.

\subsection{Contributions}

Our main contributions are as follows.
\begin{itemize}
    \item We prove that SRPT-$n$ is weakly tail optimal\added{, but not strongly tail optimal,} for the M/G/$n$ queue with heavy-tailed job sizes. (\removed{Appendix~\ref{sec:SRPT-n weak tail optimal new}}\added{Appendices~\ref{sec:SRPT-n weak tail optimal new} and~\ref{sec:SRPT-n not strongly tail optimal}})
    \item We introduce \SPLIT, the first scheduling policy that is strongly tail optimal for the M/G/$n$ queue with heavy-tailed job sizes. \SPLIT achieves strong tail optimality when $\rho < \frac{n-1}{n}$, and empirically outperforms SRPT-$n$ consistently across 
    high response times. (Section~\ref{sec:SPLIT})
    \item We introduce $\threshsplit{d}$, a parameterized family of policies. 
    By choosing the threshold $d$ appropriately, $\threshsplit{d}$ can be made arbitrarily close to strongly tail optimal under any load $\rho < 1$. 
    (Section~\ref{sec:threshsplit})
    \item We generalize $\threshsplit{d}$ to $\TAGsplit{d}$ for the setting where exact job sizes are unknown to the scheduler. Relying on a standard approximation in the TAGS literature (Approximation~\ref{approx:TAGS}), $\TAGsplit{d}$ can likewise be made arbitrarily close to strongly tail optimal under any load $\rho < 1$. (Section~\ref{sec:TAG-SPLIT})
\end{itemize}

%% file: 2-prior.tex
\section{Prior Work}
\label{sec:prior}

We now review the most relevant prior work, organized into three areas:
scheduling policies that optimize mean response time (Section~\ref{sec:prior:mean}),
analysis of tail behavior under various scheduling disciplines (Section~\ref{sec:prior:tail-analysis}),
and policies designed to optimize tail latency (Section~\ref{sec:prior:tail-opt}).

\subsection{Scheduling to optimize mean response time}
\label{sec:prior:mean}

\removed{The most extensively studied objective in scheduling theory is minimizing mean response time.} In the single-server M/G/1 queue, SRPT is known to minimize mean response time \cite{schrage1968proof}. When exact job sizes are unknown, the Gittins index policy is mean-optimal \cite{scully2021gittins,aalto2009gittins}.
In the case of heavy traffic, these single-server results extend to the multi-server M/G/$n$; specifically, SRPT-$n$ and M-Gittins are asymptotically mean-optimal in heavy traffic \cite{grosof2019srpt,scully2020optimal}.

Recently, Grosof and Hurtado-Lange \cite{grosof2025outperforming} introduced the SEK policy (see Section \ref{sec:simulation} for more discussion),  which is the only prior work observing that giving some priority to large jobs can benefit mean response time. \removed{The authors observe that in a few particular states where there are exactly $n + 1$ jobs of particular sizes, one doesn't always want to run the $n$ jobs of shortest remaining time.  They build a policy, SEK, that makes exceptions in these particular states, and prove that their policy outperforms}\added{They prove that SEK just slightly outperforms} SRPT-$n$ for mean response time at all loads.

Our work differs fundamentally in objective: we optimize the \emph{tail} of the response time rather than the mean. As established by Wierman and Zwart \cite{wierman2012tail}, tail optimization is a fundamentally different problem from mean optimization, requiring qualitatively different scheduling strategies.

\subsection{Analysis of tail behavior}
\label{sec:prior:tail-analysis}

A substantial body of work has analyzed the \emph{tail} of the response time distribution under various scheduling policies. These works focus on characterizing tail behavior rather than establishing optimality results.

\paragraph{Single-server results.}
In the M/G/1 with regularly varying service times of index $(-\alpha)$, the response time tail depends critically on the scheduling discipline. FCFS yields a \emph{heavier} tail of index $-(\alpha-1)$, since a single large job delays all subsequent arrivals \cite{boxma2007tails,borst2003impact}. Disciplines that protect small jobs from such blocking, including PS, FB, and SRPT, instead preserve the original service time index $(-\alpha)$ \cite{zwart2000sojourn,borst2006sojourn,guillemin2004tail,nuyens2008preventing,nunezqueija2002equally}, and the full range of achievable indices is characterized in \cite{boxma2011sojourn}. Under light-tailed job sizes the ordering reverses: the response time tail decays fastest under FCFS \cite{abate1995exponential}, while SRPT produces the \emph{worst} tail rather than the best \cite{nuyens2006large}.


\paragraph{Multi-server results.}
The multi-server case is far less understood. To the best of our knowledge, all existing tail-analytic results concern FCFS, covering the two-server queue and its heterogeneous and low-load variants \cite{foss2006heavy,foss2012large,boxma2002mg2,blanchet2016tail}.


Notably, all of these results analyze specific policies without addressing optimality, and the tail behavior of non-FCFS policies in multi-server systems remains entirely unexplored prior to the current paper.
\subsection{Scheduling to optimize tail latency}
\label{sec:prior:tail-opt}

In the previous subsection, we covered work which analyzes the tail behavior of various specific policies, that work did not identify which policies are {\em optimal} in their response time tail.  We now discuss prior work on tail-optimality. 

\paragraph{Single-server results.}
For {\em heavy-tailed} job sizes, the setting that is the focus of this paper, single-server tail optimality is already characterized in Section~\ref{sec:intro:summary} (see Table~\ref{tab:single-server}), so we do not repeat it here. The {\em light-tailed} case was only mentioned briefly in the introduction, so we expand on it now. Boxma and Zwart \cite{boxma2007tails} showed that FCFS is weakly tail optimal. Grosof et al.\ \cite{grosof2021nudge} introduced Nudge, which stochastically improves upon FCFS, demonstrating that FCFS is not strongly tail optimal. Van Houdt \cite{vanhoudt2022nudge} and Charlet and Van Houdt \cite{charlet2024nudgem} built upon Nudge, further improving tail performance. Only recently, Yu and Scully \cite{yu2024strongly} proposed Boost, the first policy achieving \emph{strong} tail optimality under light-tailed distributions. Harlev et al.\ \cite{harlev2025gittins} extended strong tail optimality to unknown job sizes via a Gittins-based policy.

Importantly, none of these policies, whether designed for heavy-tailed or light-tailed distributions, prioritizes large jobs. The optimal policies for heavy-tailed distributions (e.g., SRPT, SMART) all prioritize short jobs. Even for light-tailed distributions where FCFS is weakly optimal, Nudge and Boost both give short jobs extra priority over large ones to improve upon FCFS.

\paragraph{Multi-server results.}
While the landscape of tail optimality is mostly clear in single-server systems, the multi-server setting remains almost entirely open. The only prior work on tail-optimal scheduling in multi-server systems is Yu et al.\ \cite{yu2025tale}, which studies {\em light-tailed} job sizes and shows that Boost is strongly tail optimal in heavy traffic. However, outside of this regime, the tail performance of Boost degrades significantly. Interestingly, the same paper proposes a heuristic called CombinedBoost that empirically achieves strong tail performance even outside of heavy traffic. In CombinedBoost, large jobs are given some priority for a better packing effect. However, no theoretical analysis is provided for CombinedBoost; all evaluations are numerical.

To the best of our knowledge, no existing policy has been proved to be strongly, or even weakly, tail optimal for multi-server systems with heavy-tailed job sizes: This paper is the first to do so.

%% file: 3-setting.tex
\section{Problem Setting and Notation}

In this section, we formally describe the system model, define the key notation,  and establish preliminary lower bounds on the response time tail that serve as benchmarks for our policies.

\subsection{Problem Setting}

\paragraph{System model.}
We consider an M/G/$n$ queueing system with $n$ identical servers. Jobs arrive according to a Poisson process with rate $\lambda$, and each job has a random size drawn independently from a distribution $S$.

\paragraph{Heavy-tailed job sizes.}
Following \cite{wierman2012tail}, we focus on \emph{regularly varying} job size distributions, i.e., distributions with Pareto-like tails. Formally, a job size distribution $S$ is regularly varying with index $\alpha > 0$ if for any $k > 0$,
\begin{equation}
    \lim_{t\to\infty} \frac{\P{S>kt}}{\P{S>t}} = k^{-\alpha}.
    \label{eq:regular varying definition}
\end{equation}
We assume $\alpha > 1$ so that the mean job size $\E{S}$ is finite.

\paragraph{Preemptible jobs.}
Jobs are preempt-resume, meaning a preempted job can later continue from where it left off with no loss of work. Preemptibility is essential in the heavy-tailed setting: \cite{anantharam1988delays} shows that no non-preemptive scheduling discipline can prevent large delays from building up under heavy-tailed job sizes (see also \cite{boxma2007tails} for a survey).

\paragraph{Known job sizes.}
We initially assume that job sizes are known to the scheduler upon arrival. In Section~\ref{sec:TAG-SPLIT}, we relax this assumption and introduce \TAGsplit{d}. 

\paragraph{Load and resource requirement.}
We define the \emph{load} of the system as $\rho := \frac{1}{n} \lambda \E{S}$ and the \emph{resource requirement} as $r := \lambda \E{S} = n\rho$; here $r$ represents the minimum number of servers needed to keep the system stable \cite{harchol2013performance}. We will also need the corresponding notation for 
the resource requirement contributed by jobs with size larger than $x$:
\begin{equation}
    r_{>x} := \lambda \cdot \P{S > x} \cdot \E{S \mid S > x} = \lambda \E{S \cdot \mathbf{1}\{S > x\}}.
    \label{eq:r_x}
\end{equation}

\paragraph{Objective.}
Our goal is to minimize the asymptotic tail of the response time distribution\added{, in the sense of weak and strong tail optimality defined in Section~\ref{sec:intro:tail}}. \removed{As defined in Section~\ref{sec:intro:tail}, a scheduling policy $\pi$ is \emph{weakly tail optimal} if
$
    \sup_{\pi'} \limsup_{t\to\infty} \frac{\P{T^\pi>t}}{\P{T^{\pi'}>t}} = c
$
for some finite constant $c \geq 1$, and \emph{strongly tail optimal} if additionally $c = 1$.}

\subsection{Preliminaries}
\label{sec:theory:lowerbound}

In this subsection, we first recall the known optimal tail for single-server systems with heavy-tailed job sizes. Then we use the single-server result to derive two lower bounds on the tail of any multi-server policy, each of which is tighter for a different load range.

\paragraph{Single-server optimal tail.}
For a single-server system with heavy-tailed job sizes, prior works \cite{boxma2007tails,wierman2012tail,yu2024strongly} have established that the optimal tail decay rate (achieved by policies like SRPT \cite{nuyens2008preventing}) is
\begin{equation}
    \lim_{t\to\infty}\frac{\P{T>t}}{\P{S>(1-\rho) t}} = 1.
    \label{eq:single server heavy tail}
\end{equation}

\paragraph{Lower bound 1: Super server.}
Consider a single ``super server'' with speed $n$. Since a super server can essentially mimic the behavior of $n$ unit-speed servers,  it strictly dominates any system of $n$ unit-speed servers. Therefore, for any multi-server policy $\pi$,
\begin{equation}
    \P{T^\pi > t}\geq \P{T^{\text{super}} > t},
    \label{eq:super server dominance}
\end{equation}
where $T^{\text{super}}$ denotes the response time under the optimal policy for the super server. Note that the load $\rho$ is the same in both the M/G/$n$ and the super server, since the super server has speed $n$ and effectively sees jobs of size $S/n$.  Applying \eqref{eq:single server heavy tail} with effective job size $S/n$ yields
\begin{equation}
    \lim_{t\to\infty}\frac{\P{T^\pi >t}}{\P{\frac{S}{n}>(1-\rho) t}} \geq \lim_{t\to\infty}\frac{\P{T^{\text{super}} >t}}{\P{\frac{S}{n}>(1-\rho) t}} = 1.
    \label{eq:lower bound super server}
\end{equation}

\paragraph{Lower bound 2: Job size.}
Since the response time is at least the job size, for any multi-server policy $\pi$ and any $t>0$,
\begin{equation}
    \frac{\P{T^\pi >t}}{\P{S>t}} \geq 1.
    \label{eq: lower bound size}
\end{equation}

\paragraph{Comparing the two lower bounds.}
Lower bound~1 from \eqref{eq:lower bound super server} has a larger denominator (and is therefore tighter) when $\rho > \frac{n-1}{n}$. Lower bound~2 from \eqref{eq: lower bound size} is tighter when $\rho < \frac{n-1}{n}$.

%% file: 4-SPLIT.tex
\section{\SPLIT}
\label{sec:SPLIT}

In this section, we introduce the \SPLIT (SymPathy for Large jobs Improves Tail latency) policy, which is strongly tail optimal when the load $\rho$ satisfies $\rho < \frac{n-1}{n}$.   For the case when $\rho \geq \frac{n-1}{n}$, we will show a different policy in Section \ref{sec:threshsplit}.

\subsection{Policy \SPLIT Definition}

\begin{wrapfigure}{r}{0.45\textwidth}
    \centering
    \includegraphics[width=0.4\textwidth]{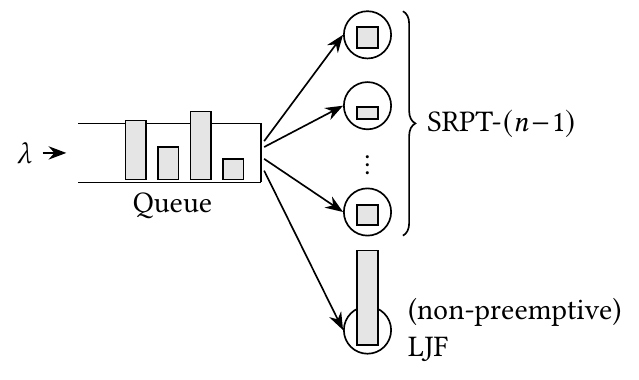}
    \caption{An illustration of the \SPLIT policy. Jobs arrive to a central queue. The LJF server always fetches the largest job; the remaining $n-1$ SRPT servers serve the jobs with the smallest remaining sizes.}
    \label{fig:SPLIT_diagram}
\end{wrapfigure}

We now define the \SPLIT policy.  See Figure~\ref{fig:SPLIT_diagram} for an illustration.

\begin{definition}[$\SPLIT$]
    \label{def:SPLIT}
    The policy \SPLIT splits $n$ servers into two groups: The first $n-1$ servers process jobs according to SRPT-$(n-1)$ and are called the \emph{SRPT servers}, and the last one server serves according to Non-preemptive Largest Job First (LJF) and is called the \emph{LJF server}. Specifically, whenever the LJF server is idle, it fetches the job with the largest ({\em original}) size in the system (either in the queue or being served by the SRPT servers) and serves it. If there are no jobs in the system, the LJF server will serve the next arriving job. 
    At every moment of time, among the rest of the jobs, the $n-1$ jobs with the smallest {\em remaining size} run on the SRPT servers. 
\end{definition}


\subsection{Main Theorem and Proof Sketch}
\label{sec:split:main}

In this subsection, we provide the main theorem, Theorem \ref{thm: SPLIT tail optimal}, and a high-level proof outline (we provide the results needed to complete the proof in Section \ref{sec:SPLIT proof small jobs} and \ref{sec:SPLIT proof large jobs}). The argument that \SPLIT is strongly tail optimal is an immediate corollary of Theorem \ref{thm: SPLIT tail optimal}, via Lower Bound 2 in Section \ref{sec:theory:lowerbound}.  



\begin{theorem}
    \label{thm: SPLIT tail optimal}
    Suppose the load $\rho$ satisfies $\rho < \frac{n-1}{n}$. Then the asymptotic response time tail under policy \SPLIT is captured by the following equation: 
    \begin{equation}
        \lim_{t\to \infty} \frac{\P{T^{\SPLIT} >t}}{\P{S>t}} = 1.
        \label{eq:SPLIT tail}
    \end{equation}
\end{theorem}

\begin{proof}

    We prove the theorem by using a tagged job argument. Consider a tagged job with size $x$, and define its response time under policy \SPLIT to be $T_x^{\SPLIT}$. We now consider three cases for $x$ (small, large, and extra large) and produce a different bound on $\P{T_x^{\SPLIT} > t}$ in each case.

    \paragraph{Case 1: Small jobs ($x \leq \frac{t}{C}$ for some constant $C>0$).}

    For these small jobs, we will show it is very unlikely that the job's response time exceeds $t$. In fact, we can prove that the probability that the job's response time exceeds $t$ decays faster than any polynomial function of $t$. Mathematically, we prove the following inequality in Lemma \ref{lemma: small jobs response time bounded}: For any $\beta>0$, there exists a constant $C>0$ such that
    \begin{equation}
        \P{T_x^{\SPLIT} > t \ \Big | \ x \leq \frac{t}{C}} = o(t^{-\beta}).
        \label{eq:SPLIT tail case 1}
    \end{equation}
    Specifically, we pick a constant $\beta>\alpha$, where $\alpha$ is defined in \eqref{eq:regular varying definition}, and take $C$ to be the corresponding constant from Lemma~\ref{lemma: small jobs response time bounded}.

    \paragraph{Case 2: Large jobs ($\frac{t}{C} < x \leq t - \sqrt{t}$).}

    For these large jobs, we show that they are very likely to end up being served by the (non-preemptive) LJF server before waiting for $\sqrt{x}$ time\footnote{Here $\sqrt{x}$ is not crucial and can be any other function in $o(x)$.}. Mathematically, define $p(x)$ to be the probability that the tagged job gets served by the LJF server before waiting for $\sqrt{x}$ time. Then the probability that the tagged job has response time larger than $t$ is bounded by $1-p(x)$, i.e., 
    \begin{align}
        &\P{T_x^{\SPLIT}>t \ \Big | \  \frac{t}{C} < x \leq t - \sqrt{t} } \nonumber\\
        &\leq \P{T_x^{\SPLIT}>x + \sqrt{x} \ \Big | \ \frac{t}{C} < x \leq t - \sqrt{t}} \nonumber\\
        &\leq 1 - \inf_{\frac{t}{C} < x \leq t - \sqrt{t}}p(x). \label{eq:SPLIT tail case 2}
    \end{align}
    
    Lemma \ref{lemma: p(x) converges to 1} proves that this probability, $p(x)$, converges to 1 as $x$ goes to infinity, and hence the tail probability in \eqref{eq:SPLIT tail case 2} goes to $0$.

    \paragraph{Case 3: Extra large jobs ($x > t - \sqrt{t}$).}
    For these extremely large jobs, we bound the probability that the tagged job has response time larger than $t$ by 1, i.e.,
    \begin{equation}
        \P{T_x^{\SPLIT}>t \ \Big | \ x > t - \sqrt{t}} \leq 1.
    \end{equation}

    Now we combine the bounds for the above three cases, dividing them by $\P{S>t}$ and taking the limit as $t$ goes to infinity: First, we have that

    \begin{align}
        \lim_{t\to \infty} \frac{\P{T^{\SPLIT}>t}}{\P{S>t}} 
        &\leq \lim_{t\to \infty} \frac{1}{\P{S>t}} \Bigg( \P{T_x^{\SPLIT}>t \ \Big | \ x \leq \frac{t}{C}} \P{S \leq \frac{t}{C}} \nonumber\\&\quad + \P{T_x^{\SPLIT}>t \ \Big | \ \frac{t}{C} < x \leq t - \sqrt{t}} \P{\frac{t}{C} < S \leq t - \sqrt{t}} + \P{S>t - \sqrt{t}} \Bigg). \label{eq: sec 4.1 main ineq}
    \end{align}

    Then by \eqref{eq:SPLIT tail case 1} we have that the first term is
    \[\lim_{t\to \infty} \frac{\P{T_x^{\SPLIT}>t \ \Big | \ x \leq \frac{t}{C}} \P{S \leq \frac{t}{C}}}{\P{S>t}} \leq  \lim_{t\to\infty} \frac{o(t^{-\beta})}{\P{S>t}} = 0.\]

    By \eqref{eq:SPLIT tail case 2} and Lemma \ref{lemma: p(x) converges to 1}, we have that the second term is
    \begin{align*}
        &\lim_{t\to \infty}\frac{1}{\P{S>t}} \P{T_x^{\SPLIT}>t \ \Big | \ \frac{t}{C} < x \leq t - \sqrt{t}} \P{\frac{t}{C} < S \leq t - \sqrt{t}} \\
        &\leq \lim_{t\to \infty}\frac{1}{\P{S>t}} \cdot \P{T_x^{\SPLIT}>t \ \Big | \ \frac{t}{C} < x \leq t - \sqrt{t}} \P{S> \frac{t}{C}} \\
        &\leq \lim_{t\to\infty} \frac{\P{S> \frac{t}{C}}}{\P{S>t}} \cdot \left(1 - \inf_{\frac{t}{C} < x \leq t - \sqrt{t}}p(x)\right) \\
        &= C^\alpha \cdot 0 = 0.
    \end{align*}

    And the third term, by using the property of regularly varying distribution \eqref{eq:regular varying definition}, is
    \[\lim_{t\to\infty} \frac{\P{S>t - \sqrt{t}}}{\P{S>t}} =1.\]

    Substituting all three terms above into \eqref{eq: sec 4.1 main ineq} gives the proof.

\end{proof}

\begin{corollary}
    \label{corollary: SPLIT strongly tail optimal}
    \SPLIT is strongly tail optimal if the load $\rho$ satisfies $\rho < \frac{n-1}{n}$. 
\end{corollary}
\begin{proof}
    For any policy $\pi'$, as stated in \eqref{eq: lower bound size}, we have that
    $ \lim_{t\to \infty} \frac{\P{T^{\pi'} >t}}{\P{S>t}} \geq 1. $
    Hence by Theorem \ref{thm: SPLIT tail optimal}, we have that 
    \[ \lim_{t\to \infty} \frac{\P{T^{\SPLIT} >t}}{\P{T^{\pi'}>t}} = \lim_{t\to \infty} \frac{\P{T^{\SPLIT} >t}}{\P{S>t}} \cdot \frac{\P{S>t}}{\P{T^{\pi'}>t}} \leq 1. \]
\end{proof}

\subsection{Proof for the bound for small jobs}
\label{sec:SPLIT proof small jobs}

In this subsection, our main goal is to prove Lemma~\ref{lemma: small jobs response time bounded} that is used in the proof of Theorem \ref{thm: SPLIT tail optimal}.
\begin{lemma}
    \label{lemma: small jobs response time bounded}
    Suppose $\rho < \frac{n-1}{n}$. Then for any $\beta>0$,
    there exists a constant $C>0$ such that
    \begin{equation}
        \P{T_x^{\SPLIT} > t \ \Big | \ x \leq \frac{t}{C}} = o(t^{-\beta}).
    \end{equation}
\end{lemma}

We start with some rough intuitions on Lemma~\ref{lemma: small jobs response time bounded}. While Lemma~\ref{lemma: small jobs response time bounded} holds for any job with size $x\leq \frac{t}{C}$, it implicitly says different things for different job sizes $x$: 
(1) For those relatively large jobs (with size $x \approx \frac{t}{C}$), the lemma says that their response time is very unlikely to be larger than a constant times $x$; (2) For smaller jobs (with size $x \ll \frac{t}{C}$), their response time is likely to be smaller than the response time of 
 relatively large jobs. Property (1) is guaranteed by the $n-1$ SRPT servers under \SPLIT -- a similar property is proved in Lemma~\ref{lemma: SRPT-n response time bounded new}. Property (2), in contrast, holds automatically under SRPT-$n$ (where there is monotonicity of response time with respect to job size), but not for \SPLIT by design. In proving Lemma~\ref{lemma: small jobs response time bounded}, we overcome this difficulty by bounding the response time under \SPLIT via relevant work, and then leveraging the monotonicity of relevant work with respect to job size. 

To rigorously prove this lemma, we proceed through a sequence of reductions (shown in Figure~\ref{fig:three_systems}), relating the response time of the tagged job under \SPLIT first to its counterpart in a single super-server system with speed $n-1$ running SRPT, and then to a single-server system with speed $1$ running SRPT (which we call SRPT-1), where standard tail bounds apply. 


\begin{figure}[h]
    \centering
    \includegraphics[width=0.7\textwidth]{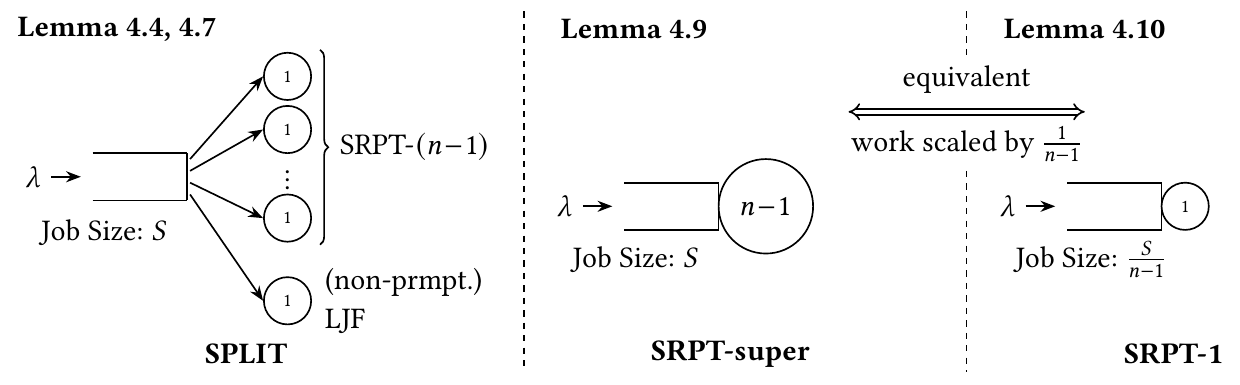}
    \caption{Three systems used in the proof: the $n$-server system under SPLIT, a single super-server with speed $n-1$ under SRPT, and a single server with speed $1$ under SRPT.}
    \label{fig:three_systems}
\end{figure}

We first define the following two notations for work.
\begin{definition}[Relevant work under \SPLIT]
    \label{def:relevant work SPLIT}
    Define the relevant work at time $t$, $W_x^{\SPLIT}(t)$, to be the total work comprising jobs in the system at time $t$ with remaining size smaller than $x$, excluding the job being served by the LJF server. Further define $W_x^{\SPLIT}$ to be the steady-state relevant work.
\end{definition}

\begin{definition}[$\WA{x}{t}$]
\label{def:WA}
    Define $\WA{x}{t}$ to be the total amount of work comprising  jobs arriving from time 0 to time $t$ with size smaller than $x$. 
\end{definition}

We can now bound the response time of a job of size $x$ by the above two types of work.  
\begin{lemma}
    \label{lemma: response time bounded by work}
    Without loss of generality, assume a tagged job of size $x$ arrives at time 0.
    For this tagged job, we have that
    \begin{equation}
        \P{T_x^{\SPLIT} > t}\leq \P{x + \frac{1}{n-1} W_x^{\SPLIT} + \frac{1}{n-1}\WA{x}{t}>t}.
        \label{eq:response time bounded by work}
    \end{equation}
\end{lemma}
\begin{proof}
    We will upper bound the work that needs to get done before the tagged job leaves the system.
    By PASTA (Poisson Arrivals See Time Averages\added{, e.g., Section 13.3 in \cite{harchol2013performance}}), \removed{when the tagged job arrives, it sees $W_x^{\SPLIT}$ relevant work in the system.}\added{the relevant work that the tagged job sees upon arrival, not including the tagged job itself, has the distribution of the steady-state relevant work $W_x^{\SPLIT}$. } Intuitively speaking, this relevant work, together with newly arriving work comprising jobs with size smaller than $x$, must be completed with rate at least $(n-1)$. The only exception is when there are fewer than $(n-1)$ jobs with remaining size smaller than the tagged job, in which case the tagged job is served. The response time of the tagged job can be upper-bounded by the above argument.   

    To see this more precisely, we define a potential function $G(s)$ at time $s$ to be the sum of three terms: (1) remaining size of the tagged job, (2) $\frac{1}{n-1}$ times the remaining work from the initial relevant work $W_x^{\SPLIT}$, and (3) $\frac{1}{n-1}$ times the remaining work from the newly arrived work $W_A(x,s)$. 
    
    
    Hence, we have that $G(0)= x + \frac{1}{n-1} W_x^{\SPLIT}$.
    The dynamics of $G(s)$ are given by: 
    At arrivals, $G(s)$ increases by $\frac{1}{n-1}$ times the size of the arriving job if the arriving job is smaller than $x$. \removed{Otherwise}\added{Between arrivals}, $G(s)$ decreases with rate at least 1 \removed{if $G(s) > 0$}\added{as long as the tagged job is in the system} (if the tagged job is not served, then all $n-1$ SRPT servers are serving jobs with size smaller than $x$, each reducing the potential function with rate $\frac{1}{n-1}$; otherwise if the tagged job is served, then the first term of the potential function already decreases with rate 1). 
    Mathematically, \added{integrating the decrease and adding the increases at arrivals,} we have that \added{for any $s$ before the tagged job completes,}
    \begin{removedblock}
    \begin{equation*}
        G(0) = x + \frac{1}{n-1} W_x^{\SPLIT},\qquad \frac{dG(s)}{ds} \leq - 1 +  \frac{1}{n-1}\lambda \P{S\leq x} \E{S\mid S\leq x}\quad \text{if $G(s)>0$}.
    \end{equation*}
    \end{removedblock}
    \begin{addedblock}
    \begin{equation}
        G(s) \leq x + \frac{1}{n-1} W_x^{\SPLIT} + \frac{1}{n-1}\WA{x}{s} - s.
        \label{eq:potential G(0) SPLIT}
    \end{equation}
    \end{addedblock}

    Consider the event that the tagged job is not completed before time $t$ (which means its response time is more than $t$). In this case, \removed{the potential function $G(s)>0$ for all $s\in(0,t]$, which means the dynamics of $G(s)$ given by \eqref{eq:potential G(0) SPLIT} always hold for any $s\in(0,t]$}\added{$G(t)>0$ since the tagged job still has positive remaining size, and \eqref{eq:potential G(0) SPLIT} holds for $s=t$}. Hence, we have that
    \begin{equation}
        0< G(t) \leq x + \frac{1}{n-1} W_x^{\SPLIT} + \frac{1}{n-1}\WA{x}{t} - t.
    \end{equation}

    Hence we have that 
    \[\P{T_x^{\SPLIT} > t}\leq \P{x + \frac{1}{n-1} W_x^{\SPLIT} + \frac{1}{n-1}\WA{x}{t}>t}. \]
    %
%
\end{proof}

Now we relate the relevant work under \SPLIT to the relevant work under SRPT-super.

\begin{definition}[Relevant work under SRPT-super]
    \label{def:relevant work}
    Define the relevant work in the SRPT-super system at time $t$, $\Wsuper{x}(t)$, to be the total work comprising jobs in the system at time $t$ with remaining size smaller than $x$. Further define $\Wsuper{x}$ to be the steady-state relevant work.
\end{definition}

\begin{lemma}
    \label{lemma: W_x SPLIT bounded}
    We have that under the same sample path, at any moment of time $t$,
    \begin{equation}
        W_x^{\SPLIT}(t) \leq \Wsuper{x}(t) + (n-1)x.
    \end{equation}
\end{lemma}

\begin{proof}
    The proof is motivated by \removed{Lemma 2.2}\added{Lemma~2 (Section~5.2)} in \cite{grosof2019srpt}. Consider the $n$-server system under policy \SPLIT. Define a ``few-jobs interval" to be a time interval where there are no more than $(n-1)$ jobs with remaining size smaller than $x$ (excluding the job being served by the LJF server) in the $n$-server system under policy \SPLIT. Define ``many-jobs interval'' to be the opposite. Then, for any time $t$ within a few-jobs interval, we have $W_x^{\SPLIT}(t)\leq (n-1) \cdot x$, and the lemma holds obviously. 

    For any time $t$ within a many-jobs interval, we argue that the difference $W_x^{\SPLIT}(t) - \Wsuper{x}(t)$ is non-increasing. This is because during many-jobs intervals, the $(n-1)$ SRPT servers in \SPLIT work on the relevant work, and the relevant work decreases with rate at least $(n-1)$, which is no slower than the super server completes the relevant work in the super-server system. Throughout, the jobs served on the LJF server are excluded in both the relevant work and the determination of a "few-jobs interval", so only arrivals start a "many-jobs interval" or increase the relevant work. 
    Moreover, the arrival sequences in both systems are the same, so arrivals do not affect the difference.
\end{proof}

We need one more lemma before we can prove Lemma \ref{lemma: small jobs response time bounded}. Lemma~\ref{lemma: single server work bound} is about a single server of speed 1 running SRPT (we refer to this as an SRPT-1 system), and slightly strengthens a lemma in \cite{nuyens2008preventing}. The lemma in \cite{nuyens2008preventing}  says that, for any $\beta>0$, there exists a constant $C$ such that
$\P{T_x^{SRPT-1} > Cx} = o(x^{-\beta}).$
\begin{lemma}[SRPT-1 work bound]
    \label{lemma: single server work bound}
    For a single-server (speed 1) system under SRPT with load $\rho<1$, we have that for any $\beta>0$, there exists a constant $C>0$ such that
    \begin{equation}
        \P{2x + \Wstar{x} + \WAstar{x}{Cx} > Cx} = o(x^{-\beta}),
    \end{equation}
where $\Wstar{x}$ refers to the (relevant) work in the SRPT-1 system consisting of jobs of remaining size $\leq x$ and $\WAstar{x}{Cx}$ refers to the total relevant arriving work into the SRPT-1 system by time $Cx$.
\end{lemma}
\begin{proof}
    See Lemma \ref{lemma: single server SRPT-1 bound} in the Appendix.
\end{proof}

We are now ready to prove Lemma \ref{lemma: small jobs response time bounded}. 

\begin{proof}[Proof of Lemma \ref{lemma: small jobs response time bounded}]
    Our goal is to bound $T_x^{SPLIT}$ where $x \leq \frac{t}{C}$ is the size of a  tagged job.  We choose to let 
    $C$ be the constant from Lemma \ref{lemma: single server work bound}. 
    
    By Lemma \ref{lemma: response time bounded by work},
    \begin{equation}
        \P{T_x^{\SPLIT} > t} \leq \P{x + \frac{1}{n-1} W_x^{\SPLIT} + \frac{1}{n-1}\WA{x}{t} > t}.
    \end{equation}
    By Lemma \ref{lemma: W_x SPLIT bounded}, $W_x^{\SPLIT} \leq_{st} \Wsuper{x} + (n-1)x$. 
    Moreover, by definition, $\Wstar{x} = \frac{1}{n-1}\Wsuper{x}$ and $\WAstar{x}{t} = \frac{1}{n-1}\WA{x}{t}$.
    Hence we have that 
    \begin{align*}
    & \P{T_x^{\SPLIT} > t \ \Big | \ x \leq \frac{t}{C}} \\
        &\leq \P{x + \frac{1}{n-1} W_x^{\SPLIT} + \frac{1}{n-1}\WA{x}{t} > t\ \Big | \ x \leq \frac{t}{C}} && \text{by Lemma \ref{lemma: response time bounded by work}} \\
        &\leq \P{x + \frac{1}{n-1} (\Wsuper{x} + (n-1)x) + \frac{1}{n-1}\WA{x}{t} > t\ \Big | \ x \leq \frac{t}{C}} && \text{by Lemma \ref{lemma: W_x SPLIT bounded}} \\
        &= \P{2x + \Wstar{x} + \WAstar{x}{t} > t\ \Big | \ x \leq \frac{t}{C}} \\
        &\leq \P{2\left(\frac{t}{C}\right) + \Wstar{\frac{t}{C}} + \WAstar{\frac{t}{C}}{t}>t} \\
        &= o\left(\left(\frac{t}{C}\right)^{-\beta}\right)= o(t^{-\beta}). && \text{by Lemma \ref{lemma: single server work bound}}
    \end{align*}
\end{proof}


\subsection{Proof for the bound for large jobs}
\label{sec:SPLIT proof large jobs}

In this section, we prove the following lemma that is used in the proof of Theorem \ref{thm: SPLIT tail optimal}.
The lemma shows that for a sufficiently large job, it is very likely that under policy \SPLIT, the job gets served by the (non-preemptive) LJF server before waiting for $\sqrt{x}$ time.

\begin{lemma}
    \label{lemma: p(x) converges to 1}
    Under policy \SPLIT,
    define $p(x)$ to be the probability that a job with size $x$ gets served by the LJF server before waiting for $\sqrt{x}$ time. Then we have that 
    \begin{equation}
        \lim_{x\to \infty} p(x) = 1.
    \end{equation}
\end{lemma}

\begin{proof}
    For the tagged job with size $x$, we consider the following two events:
    \begin{enumerate}
        \item When the tagged job arrives, it sees no job with original size larger than $\sqrt{x}$. Define the probability of this event to be $p_1(x)$.
        \item The tagged job does not see any job with size larger than $x$ arriving during the $\sqrt{x}$ time after it arrives. Define the probability of this event to be $p_2(x)$.
    \end{enumerate}

    Note that the above two events are independent. Moreover, we know that if both of the above two events happen, the tagged job gets served by the LJF server immediately after the LJF server finishes the remaining work of the job being served (which must be smaller than $\sqrt{x}$ because event (1) happens). Hence we have that 
    \begin{equation}
        p(x)\geq p_1(x) p_2(x).
    \end{equation}
    Now we characterize $p_1(x)$ and $p_2(x)$ respectively.

    \paragraph{Characterizing $p_1(x)$.}
    Because the arrival process is Poisson, by the PASTA theorem (e.g., Section 13.3 in \cite{harchol2013performance}), the probability that the  tagged job sees no job with original size $>\sqrt{x}$ upon arrival is equal to the time average probability that the system does not have jobs with original size  $> \sqrt{x}$. 
    
    We define ``red jobs'' to be the jobs with original size larger than $\sqrt{x}$.
    Define $W_{red}(t)$ to be the total remaining work in the system from red jobs at time $t$. Then by the above PASTA argument:
    \[p_1(x) = \lim_{t\to\infty} \frac{\int_{0}^t \mathbf{1}[W_{red}(s)=0] ds}{t}.\]

    We know that under policy \SPLIT, if there exists a red job in the system, either: (1) The LJF server is serving some red job; or (2) The LJF server will start serving red jobs after it finishes the remaining work of the job being served (if at that time red jobs still exist). Moreover, if the LJF server starts serving red jobs, it will keep serving red jobs until there are no red jobs in the system. Hence, suppose at time $t$ a red job arrives and $W_{red}(t)>0$ (while $W_{red}(t-)=0$). Then either $W_{red}$ immediately decreases with rate at least 1 between red arrivals until it drops to 0 (note that during this period, there could be also SRPT servers serving red jobs), or $W_{red}$ starts decreasing with rate at least 1 after a time delay of at most $\sqrt{x}$ (because the LJF server is currently serving a job with original size at most $\sqrt{x}$).  See Figure \ref{fig:redwork} for an illustration. 

    \begin{figure}[h]
        \centering
        \includegraphics[width=0.55\textwidth]{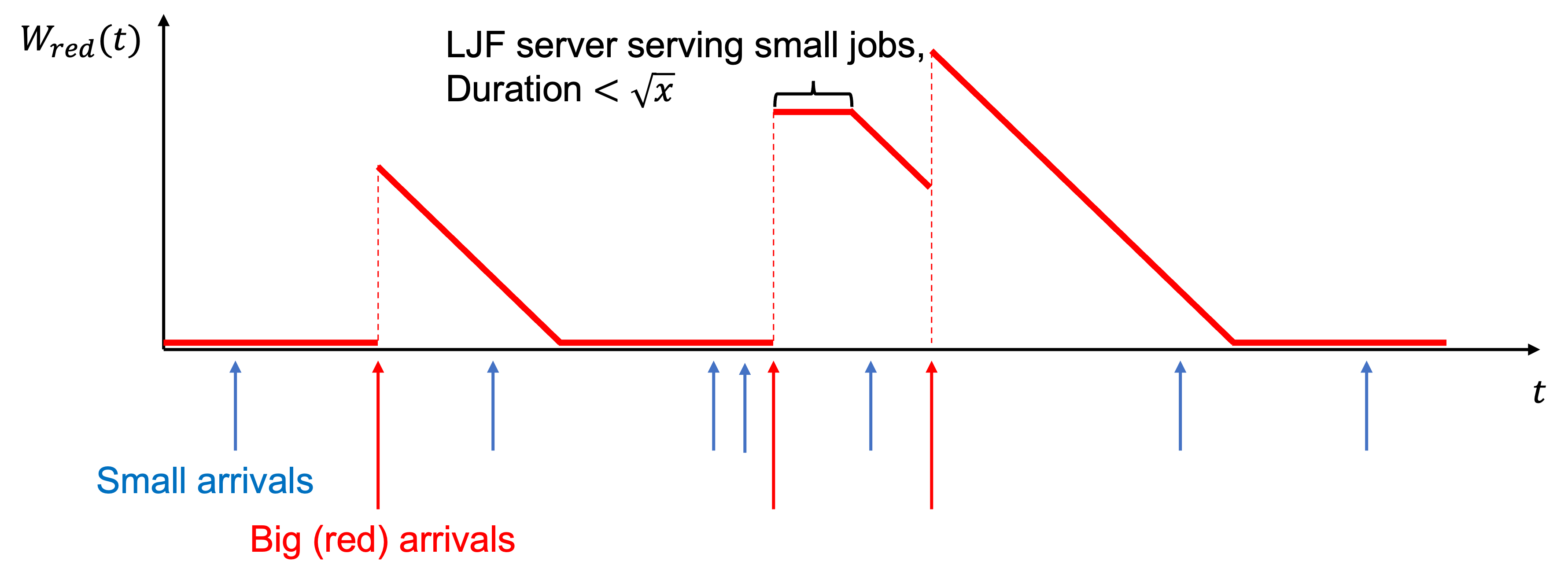}
        \caption{The dynamic of $W_{red}(t)$ under policy \SPLIT.}
        \label{fig:redwork}
    \end{figure}
    
    During each red work period (the consecutive period when $W_{red}(t) > 0$), we know that for at most $\sqrt{x}$ time, $W_{red}$ is not decreasing; for the rest of the time, $W_{red}$ is decreasing with rate at least 1 between red arrivals. Define $B$ to be the duration of the red work period, and $W$ to be the total red work arriving during this period including the red jobs that start the red work period (clearly $W>\sqrt{x}$).  Then\added{, since the total decrease of $W_{red}$ during the period equals $W$,} 
    \[B \leq \sqrt{x} + W< 2W.\]

    Define $B^{(i)}$ to be the duration of the $i$-th red work period, which ends at time $t^{(i)}$, and $W^{(i)}$ to be the total red work arriving during the $i$-th red work period (including the starting red jobs). Then 
    \[p_1(x) =  \lim_{t\to\infty} \frac{\int_{0}^t \mathbf{1}[W_{red}(s)=0] ds}{t}= 1- \lim_{i\to\infty} \frac{\sum_{j=1}^{i} B^{(j)}}{t^{(i)}} \geq 1- \lim_{i\to\infty} \frac{2\sum_{j=1}^{i} W^{(j)}}{t^{(i)}} = 1- 2r_{>\sqrt{x}},
    \]
    where $r_{>x}$ was defined in \eqref{eq:r_x}.

    \paragraph{Characterizing $p_2(x)$.}    
    The arrival process of jobs larger than $x$ is a Poisson process with rate $\lambda \cdot \P{S>x}$. Hence, the probability that no red job  arrives during the $\sqrt{x}$ time after the tagged job arrives is given by 
    \[p_2(x) = e^{-\lambda \cdot \P{S>{x}} \sqrt{x}}.\]

    Hence, we have that 
    $p(x) \geq p_1(x) p_2(x) \geq (1-2r_{>\sqrt{x}}) e^{-\lambda \cdot \P{S>{x}} \sqrt{x}}$,
    where both terms on the right-hand side converge to 1 as $x$ goes to infinity.
\end{proof}

%% file: 5-Thresh-SPLIT.tex
\section{\threshsplit{d}: tail optimal under all loads}
\label{sec:threshsplit}

\begin{wrapfigure}[15]{r}{0.4\textwidth}
    \centering
    \includegraphics[width=0.40\textwidth]{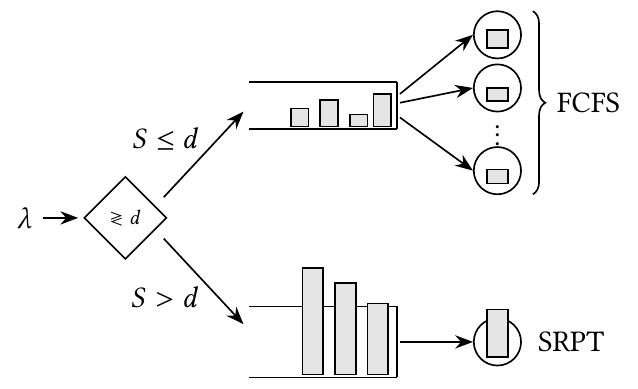}
    \captionsetup{justification=raggedright}
    \caption{An illustration of the \threshsplit{d} policy. Jobs are classified by size: small jobs ($S \leq d$) are served by $n-1$ FCFS servers, and big jobs ($S > d$) are served by a single SRPT server.}
    \label{fig:thresh_SPLIT_diagram}
\end{wrapfigure}
 
In this section, we introduce the \threshsplit{d} policy (see Figure~\ref{fig:thresh_SPLIT_diagram}).
We show that for load $\rho\geq \frac{n-1}{n}$, when the threshold parameter $d$ is properly chosen, \threshsplit{d} can be made arbitrarily close to strongly tail optimal. For completeness, we show that \threshsplit{d} achieves this for load $\rho<\frac{n-1}{n}$ as well, although \SPLIT is already strongly tail optimal in that regime. \added{In contrast, \threshsplit{d} only approaches strong tail optimality as $d$ approaches the limiting value in Theorem~\ref{thm: main}; for any fixed $d$, it is not strongly tail optimal.}


\begin{definition}[$\threshsplit{d}$]
\label{def:threshsplit}
Parameterized by a threshold $d$, the policy \threshsplit{d} classifies all jobs as either {\em big}, if their size is $>d$, or {\em small}, if their size is $\leq d$.  It then splits the $n$ servers into $n-1$ servers reserved for small jobs and 1 server reserved for big jobs. The $n-1$ small-job servers apply FCFS (with a central queue) and the 1 big-job server applies policy SRPT.
\end{definition}

The next lemma is a conventional result for the response time tail of M/G/$n$ systems with light-tailed job sizes under FCFS. 

\begin{lemma} [Theorem 2 in \cite{sadowsky1989analysis}]
    \label{lemma: light tail}
        For any FCFS M/G/$n$ system with light-tailed job sizes and load $\rho<1$, let $T$ denote the response time in this system. Then there exists a constant $\gamma>0$ such that
        \[ \lim_{t\to \infty} \frac{\P{T>t}}{e^{-\gamma t}} \leq 1.\]
    
    \end{lemma}

Lemma~\ref{lemma: light tail} is precisely what we need to handle the small-job side of \threshsplit{d}. The $n-1$ small-job servers only serve jobs with size bounded by $d$, so their size distribution is trivially light-tailed. Thus by Lemma~\ref{lemma: light tail} the response time tail of the small jobs decays exponentially. Since the overall response time tail is heavy-tailed (driven by the big jobs), this exponential decay implies that small jobs  do not contribute to the asymptotic response time tail of the system.

Moreover, since small jobs' sizes are bounded by $d$, the same exponentially decaying response time tail should hold under most reasonable scheduling policies on the $n-1$ small-job servers, e.g., SRPT-$(n-1)$. We choose FCFS in Definition~\ref{def:threshsplit} because Lemma~\ref{lemma: light tail} gives a rigorous proof in that case, whereas to our knowledge no analogous formal result is available for multi-server SRPT with bounded sizes. We revisit this design choice in Section~\ref{sec:discussion:design} and validate it empirically in Section~\ref{sec:simulation} and Appendix~\ref{sec:appendix:fcfs-vs-srpt}.

Under policy \threshsplit{d}, the load on the single big-job server is $r_{>d}$ (recall \eqref{eq:r_x}). The resource requirement on the $n-1$ small-job servers is then $$r_{\leq d} := \rho\cdot n - r_{>d}.$$  For the small-job system to be stable, we need $r_{\leq d} < n-1$, i.e., the load on the small system must be $<1$. Given this stability condition, we prove Theorem~\ref{thm: split tail}, from which we will derive our main theorem, Theorem~\ref{thm: main}, on the tail optimality of \threshsplit{d}.


\begin{theorem}
\label{thm: split tail}
    Under policy \threshsplit{d}, if $r_{\leq d} < n-1$ (the small-job system is stable), then
    \begin{equation}
        \lim_{t\to \infty} \frac{\P{T^{\threshsplit{d}}>t}}{\P{S>(1-r_{>d})t}} = 1.
        \label{eq: split tail}
    \end{equation}
\end{theorem}

\begin{proof}

\noindent\textbf{Small jobs.} Since $r_{\leq d} < n-1$, the small-job system (the $n-1$ servers serving only jobs with size $\leq d$) has load smaller than 1. By Lemma~\ref{lemma: light tail}, there exists a constant $\gamma>0$ such that
\begin{equation}
    \label{eq:proof small jobs tail}
    \lim_{t\to \infty} \frac{\P{T^{\threshsplit{d}}>t \mid S\leq d}}{e^{-\gamma t} }\leq 1.
\end{equation}

\noindent\textbf{Big jobs.} On the single big-job server, SRPT-1 achieves the asymptotic tail given by \eqref{eq:single server heavy tail}. Hence,
\begin{align}
   1&=\lim_{t\to \infty} \frac{\P{T^{\threshsplit{d}}>t\mid S>d}}{\P{S>(1-r_{>d})t \mid S>d}}\\
    &= \lim_{t\to \infty} \P{T^{\threshsplit{d}}>t\mid S>d}\cdot  \frac{\P{S>d}}{\P{S > (1-r_{>d}) t}}.
    \label{eq:proof large jobs tail}
\end{align}

\noindent\textbf{Combining.} By the law of total probability,
\begin{align*}
    &\lim_{t\to\infty} \frac{\P{T^{\threshsplit{d}}>t}}{\P{S>(1-r_{>d}) t}} \\
    &= \lim_{t\to\infty} \frac{\P{T^{\threshsplit{d}}>t \mid S\leq d} \P{S\leq d} + \P{T^{\threshsplit{d}}>t \mid S>d} \P{S> d}}{\P{S>(1-r_{>d}) t}}.
\end{align*}

The first term vanishes because the exponential decay in \eqref{eq:proof small jobs tail} is faster than any heavy-tailed distribution. The second term converges to 1 by \eqref{eq:proof large jobs tail}.
\end{proof}

\begin{theorem}
\label{thm: main}
    By choosing the threshold $d$ appropriately, \threshsplit{d} can be made arbitrarily close to strongly tail optimal. Specifically, the following asymptotic tails match the lower bounds in Section \ref{sec:theory:lowerbound}:
    \begin{itemize}
        \item If $\rho\leq\frac{n-1}{n}$, we have that
        \begin{equation}
            \lim_{d\to \infty} \left(\lim_{t\to \infty}\frac{\P{T^{\threshsplit{d}} >t}}{\P{S>t}}\right) = 1.
            \label{eq:split tail in max stability}
        \end{equation}
        \item If $\rho> \frac{n-1}{n}$, we have that
        \begin{equation}
            \lim_{d\to (d^*)^-} \left(\lim_{t\to \infty}\frac{\P{T^{\threshsplit{d}} > t}}{\P{S>n(1-\rho) t}}\right) = 1,
            \label{eq:split tail in min stability}
        \end{equation}
        where $d^*$ is the smallest threshold such that $r_{\leq d^*} = (n-1)$.  
    \end{itemize}
\end{theorem}

\begin{proof}
    Both cases follow from Theorem \ref{thm: split tail}.

    If $\rho\leq \frac{n-1}{n}$, then for any $d>0$ the small-job system has load less than 1. As $d\to\infty$, we have $r_{>d} \to 0$, so by Theorem \ref{thm: split tail}, the tail ratio converges to: 
    \[\lim_{d\to\infty}\left(\lim_{t\to \infty}\frac{\P{T^{\threshsplit{d}} >t}}{\P{S>t}}\right) = \lim_{d\to\infty} \lim_{t\to\infty} \frac{\P{S>(1-r_{>d})t}}{\P{S>t}}=\lim_{d\to\infty} (1-r_{>d})^{-\alpha} = 1.\]

    If $\rho> \frac{n-1}{n}$, stability of the small-job system requires $r_{>d} = n\rho - r_{\leq d} > n\rho - (n-1)$. As $d$ increases toward $d^*$ from below, $r_{>d}$ decreases toward $n\rho - (n-1)$, and the factor $(1-r_{>d})$ in \eqref{eq: split tail} approaches $n(1-\rho)$. Substituting into Theorem~\ref{thm: split tail} gives \eqref{eq:split tail in min stability}.
\end{proof}

%% file: 6-TAG-SPLIT.tex
\section{\TAGsplit{d}}
\label{sec:TAG-SPLIT}

In this section, we introduce the \TAGsplit{d} policy, a variant of \threshsplit{d} that achieves tail optimality even when the scheduler does not know exact job sizes upon arrival. Inspired by the Task-Assignment-by-Guessing-Size (TAGS) policy \cite{harchol2002task}, \TAGsplit{d} identifies big jobs by running all jobs for a fixed time $d$ and checking whether they complete. 

\begin{definition}[$\TAGsplit{d}$]
\label{def:TAGsplit}
The policy \TAGsplit{d}, parameterized by a threshold $d$, splits the $n$ servers into $n-1$ small-job servers operating under FCFS and 1 big-job server running Processor Sharing (PS), as shown in Figure~\ref{fig:TAGS_SPLIT_diagram}. Every arriving job initially enters the FCFS system as a small job. If a job has been served for $d$ time units without completing, it is identified as a big job and migrated to the big-job server in a preempt-resume manner (no work is lost). The PS server schedules all big jobs by giving every job an equal share of the server at every moment of time.  
\end{definition}

\begin{figure}[h]
    \centering
    \includegraphics[width=0.55\textwidth]{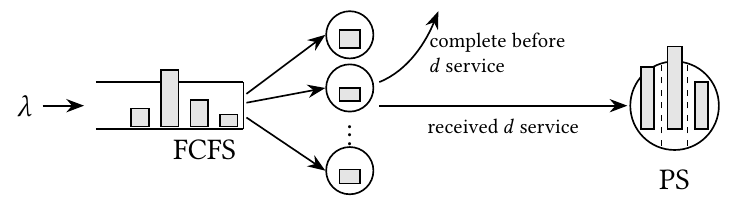}
    \caption{An illustration of the \TAGsplit{d} policy. All jobs initially enter the FCFS system. Jobs that have received $d$ service without completing are migrated to the PS server.}
    \label{fig:TAGS_SPLIT_diagram}
\end{figure}

TAGS-inspired policies rely on running jobs in upstream servers until a job has received a certain amount of service without completing; such a job is then deemed large and migrated to a downstream server. A central difficulty in the analysis of TAGS-inspired policies is that, even when the external arrivals to the upstream servers are Poisson, the arrival process to the downstream servers is neither Poisson nor renewal in general, since it is the upstream departure process thinned to jobs whose size exceeds the threshold. Most prior work in the TAGS literature has handled this difficulty by approximating the downstream arrival process as Poisson \cite{harchol2002task, BDS2019, BDS2020}. Empirical experiments have been used to support the accuracy of this approximation \cite{harchol2002task, BDS2019}. We make the analogous approximation here.

\begin{approximation}[Poisson arrivals at the big-job server]
\label{approx:TAGS}
In the analysis of \TAGsplit{d}, we assume that the arrival process to the PS (big-job) server is a Poisson process with rate $\lambda \cdot \P{S>d}$.
\end{approximation}

\added{Approximation~\ref{approx:TAGS} preserves the long-run arrival rate $\lambda\P{S>d}$ of big jobs, but it ignores the dependence among migration times created by the FCFS queue at the small-job servers. We do not have a proof that this dependence does not affect the asymptotic tail; our simulations in Section~\ref{sec:simulation:tags}, which do not use the approximation, agree with Theorem~\ref{thm: TAGsplit main}. Given Approximation~\ref{approx:TAGS}, we can now apply the well-known response time tail of the M/G/1 queue under PS \cite{guillemin2004tail}.}

For a threshold $d$, define the load of the big-job server under \TAGsplit{d} to be $\rho^{TAGS}_{large}(d)$:
\begin{equation}
    \rho^{TAGS}_{large}(d) = \lambda \cdot \P{S>d} \cdot \left(\E{S\mid S>d} - d\right).
    \label{eq:rho_TAG_large}
\end{equation}
We now state Theorem~\ref{thm: TAGsplit tail} for \TAGsplit{d}, which is analogous to Theorem \ref{thm: split tail} for \threshsplit{d}:
\begin{restatable}{theorem}{TAGsplittail}
\label{thm: TAGsplit tail}
    Under policy \TAGsplit{d}\added{ and Approximation~\ref{approx:TAGS}}, if $r_{small} := \rho\cdot n - \rho^{TAGS}_{large}(d) < n-1$ (i.e., the small-job system is stable), then
    \begin{equation}
        \lim_{t\to \infty} \frac{\P{T^{\TAGsplit{d}}>t}}{\P{S>(1-\rho^{TAGS}_{large}(d))t}} = 1.
        \label{eq:TAGsplit tail}
    \end{equation}
\end{restatable}

\noindent The proof follows the same structure as that of Theorem~\ref{thm: split tail}, splitting into small jobs (light-tailed) and big jobs (whose tail is governed by the PS server), and is deferred to Appendix~\ref{appendix:proof TAGsplit tail}.

Given Theorem~\ref{thm: TAGsplit tail}, we can now state the analogue of Theorem~\ref{thm: main} for \TAGsplit{d}:

\begin{theorem}
\label{thm: TAGsplit main}
    By choosing the threshold $d$ appropriately, \TAGsplit{d} can be made arbitrarily close to strongly tail optimal \added{under Approximation~\ref{approx:TAGS}}. Specifically, the following asymptotic tails match the lower bounds in Section \ref{sec:theory:lowerbound}:
    \begin{itemize}
        \item If $\rho\leq\frac{n-1}{n}$, we have that
        \begin{equation}
            \lim_{d\to \infty} \left(\lim_{t\to \infty}\frac{\P{T^{\TAGsplit{d}} >t}}{\P{S>t}}\right) = 1.
            \label{eq:TAGsplit tail in max stability}
        \end{equation}
        \item If $\rho> \frac{n-1}{n}$, we have that
        \begin{equation}
            \lim_{d\to (d^*)^-} \left(\lim_{t\to \infty}\frac{\P{T^{\TAGsplit{d}} > t}}{\P{S>n(1-\rho) t}}\right) = 1,
            \label{eq:TAGsplit tail in min stability}
        \end{equation}
        where $d^*$ is the smallest threshold such that $\rho \cdot n-\rho^{TAGS}_{large}(d^*) =  n-1$.
    \end{itemize}
\end{theorem}

\begin{proof}
    The proof follows the same argument as Theorem~\ref{thm: main}, with $r_{>d}$ replaced by $\rho^{TAGS}_{large}(d)$.
\end{proof}

%% file: 7-simulation.tex
\section{Simulation}
\label{sec:simulation}

In this section we empirically evaluate our policies and compare them with several baselines. Each simulation runs an M/G/$n$ system with Poisson arrivals, Pareto-distributed job sizes with $\P{S>t} = t^{-\alpha}$ for $t \geq 1$, and $10^{10}$ job arrivals.

\paragraph{Metric}
Motivated by the lower bounds in Section~\ref{sec:theory:lowerbound}, we plot the \emph{normalized tail}, i.e., the response time tail $\P{T>t}$ divided by the applicable lower bound: for $\rho<\frac{n-1}{n}$, we plot $\P{T>t}/\P{S>t}$, and for $\rho\geq \frac{n-1}{n}$, we plot $\P{T>t}/\P{S>n(1-\rho)t}$. A normalized tail converging to $1$ as $t\to\infty$ corresponds to strong tail optimality.

\paragraph{Baselines}
We compare against SRPT-$n$, and SEK-$\epsilon$; here SEK-$\epsilon$ refers to the Practical-SEK policy of \cite{grosof2025outperforming}, which mostly follows SRPT-$n$ except when the system contains exactly $n+1$ jobs with one of size larger than $\epsilon$ and the other $n$ of size smaller than $\epsilon$, in which case it serves the largest job in place of the second-largest. For SEK-$\epsilon$, the asymptotic tail empirically improves as $\epsilon$ grows, but stops improving at $\epsilon = 200$.  Hence, 
we plot only the SEK-$\epsilon$ curve with $\epsilon = 200$, which is best
among the SEK family in our experiments. We provide SEK-$\epsilon$ curves for a range of $\epsilon$ in Appendix~\ref{sec:appendix:sek-epsilon}. 
We do not include Boost \cite{yu2024strongly,yu2025tale}, which is 
only defined for light-tailed job size distributions.

\paragraph{Our policies}
We simulate \SPLIT, which applies only when $\rho<\frac{n-1}{n}$, and \threshsplit{d} for several values of $d$. In our experiments, \threshsplit{d} is implemented with a work-stealing improvement: when the large-job server is idle, it opportunistically serves small jobs, and the moment a job with size $>d$ arrives it preempts the job in service and returns to serving large jobs.

\subsection{Low load ($\rho<\frac{n-1}{n}$)}
\label{sec:simulation:lowload}

\vspace{-10pt}
\begin{figure}[h]
\centering
\begin{subfigure}[b]{0.45\textwidth}
  \includegraphics[width=\textwidth]{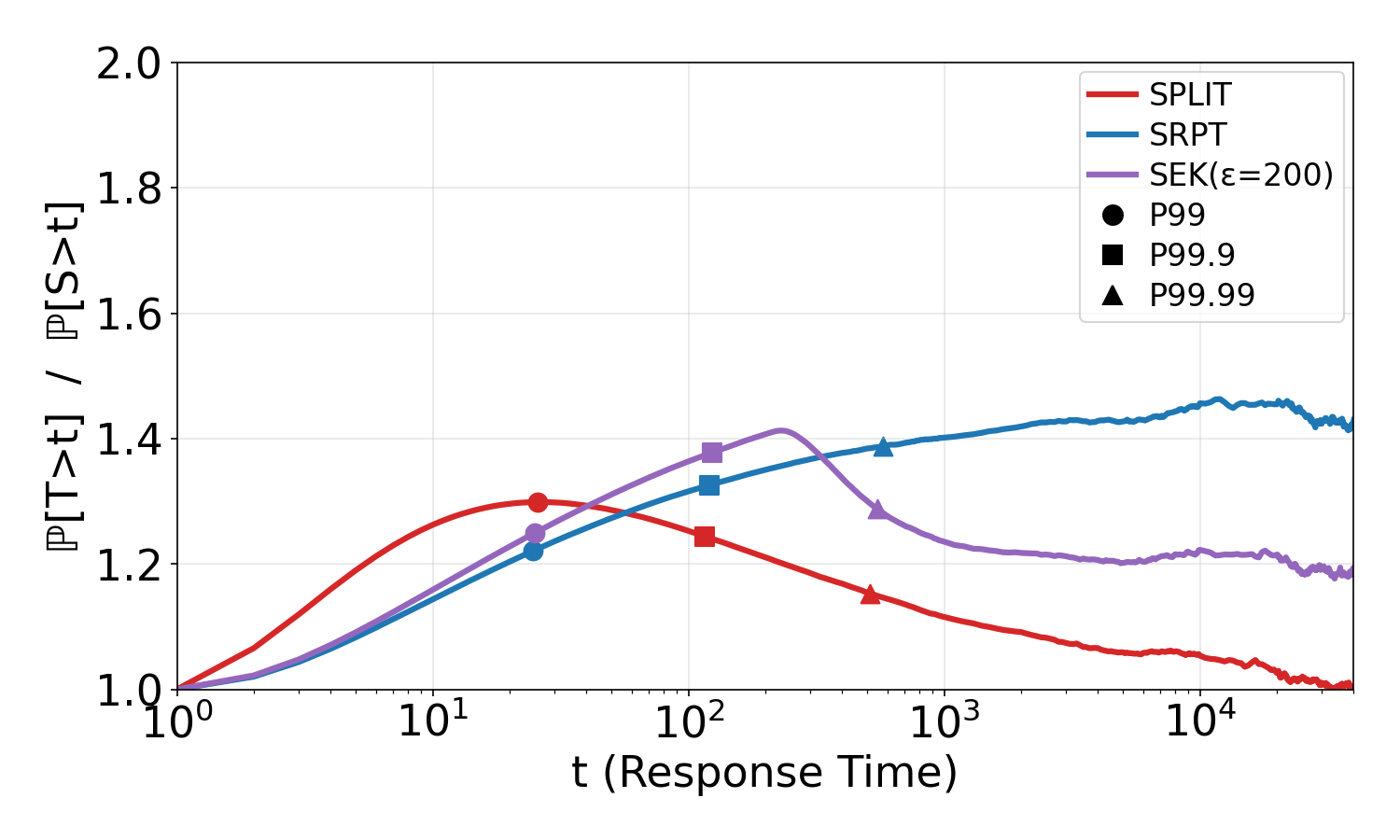}
  \caption{\SPLIT and baselines}
  \label{subfig:lowload-main}
\end{subfigure}\hfill
\begin{subfigure}[b]{0.45\textwidth}
  \includegraphics[width=\textwidth]{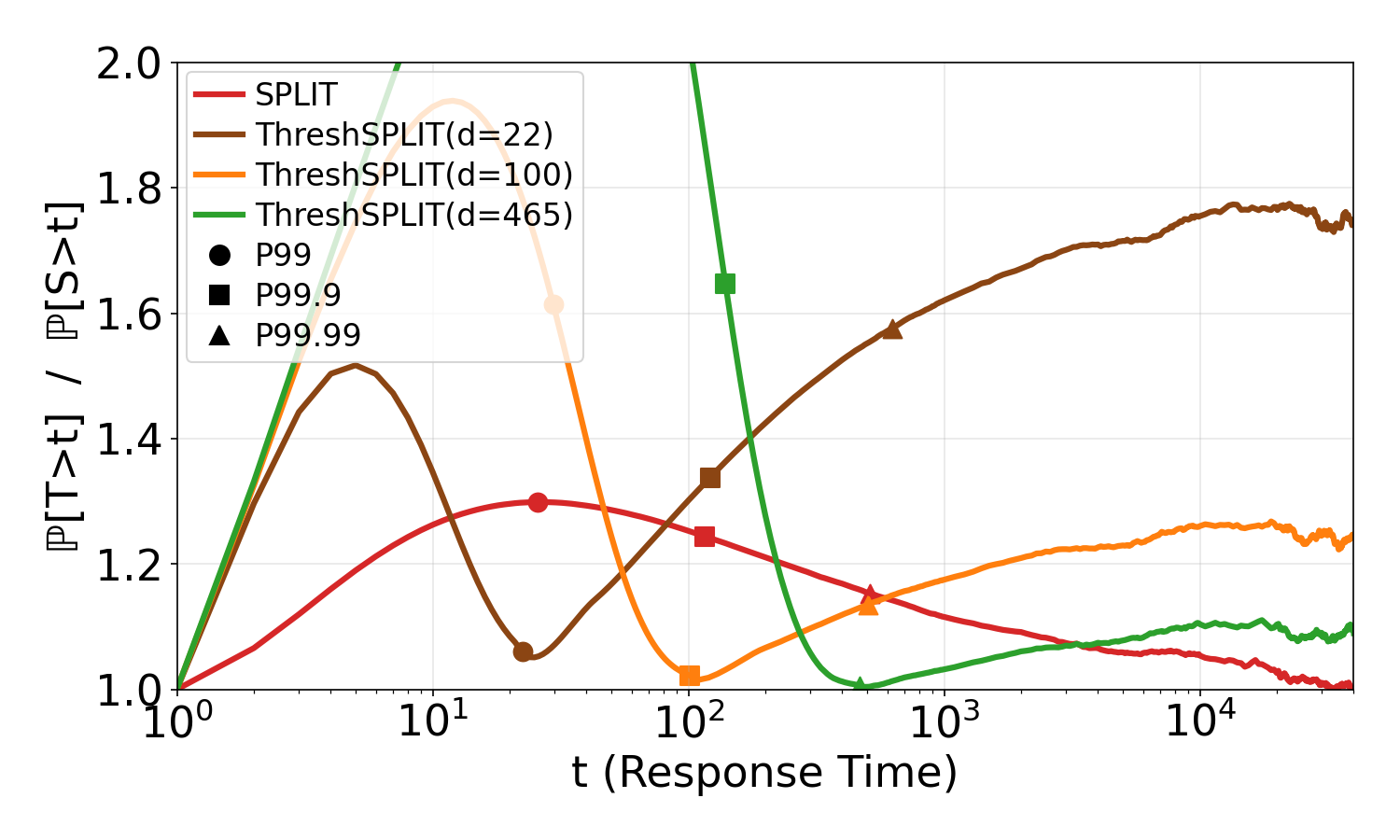}
  \caption{\threshsplit{d} for several thresholds $d$}
  \label{subfig:lowload-thresh}
\end{subfigure}
\vspace{-8pt}
\caption{Normalized response time tail, $n=3$ servers, $\rho=0.5$, Pareto job sizes with $\alpha=1.5$.}
\label{fig:sim:lowload}
\end{figure}

Figure~\ref{fig:sim:lowload} shows results for $n=3$, $\rho=0.5$, and Pareto job sizes with tail index $\alpha=1.5$ and mean $\frac{\alpha}{\alpha-1}=3$. 
Figure~\ref{subfig:lowload-main} compares \SPLIT (in red) with the baselines: \SPLIT's normalized tail converges to $1$, confirming its strong tail optimality (Theorem~\ref{thm: SPLIT tail optimal}), whereas all baselines remain bounded away from $1$. In particular, while we were able to prove that SRPT-$n$ is weakly optimal (see Appendix \ref{sec:SRPT-n weak tail optimal new}), Figure~\ref{subfig:lowload-main} \removed{clearly indicates}\added{confirms, as we prove in Appendix~\ref{sec:SRPT-n not strongly tail optimal},} that SRPT-$n$ is \removed{unlikely to be}\added{not} strongly optimal.

\SPLIT outperforms SRPT-$n$ at the $99.9\%$-tile and $99.99\%$-tile of response time, but trails SRPT-$n$ at the $99\%$-tile of response time in this setting.  The loss in the $99\%$-tile of response time is not universal: in the $\alpha=2.0$ experiment \removed{of Appendix~D.1}\added{below}, \SPLIT also wins at the $99\%$-tile of response time (see Figure~\ref{subfig:alpha2}).

Figure~\ref{subfig:lowload-thresh} reports \threshsplit{d} for three choices of the threshold $d$ corresponding to the $99\%$, $99.9\%$, and $99.99\%$-tile of the job size distribution, respectively.  Consistent with Theorem~\ref{thm: main}, the normalized tail approaches $1$ as $d$ grows. 

\begin{addedblock}
Figure~\ref{fig:sim:alpha} varies the Pareto tail index $\alpha$, fixing $n=3$ and $\rho=0.5$.

\begin{figure}[h]
\centering
\begin{subfigure}[b]{0.45\textwidth}
  \includegraphics[width=\textwidth]{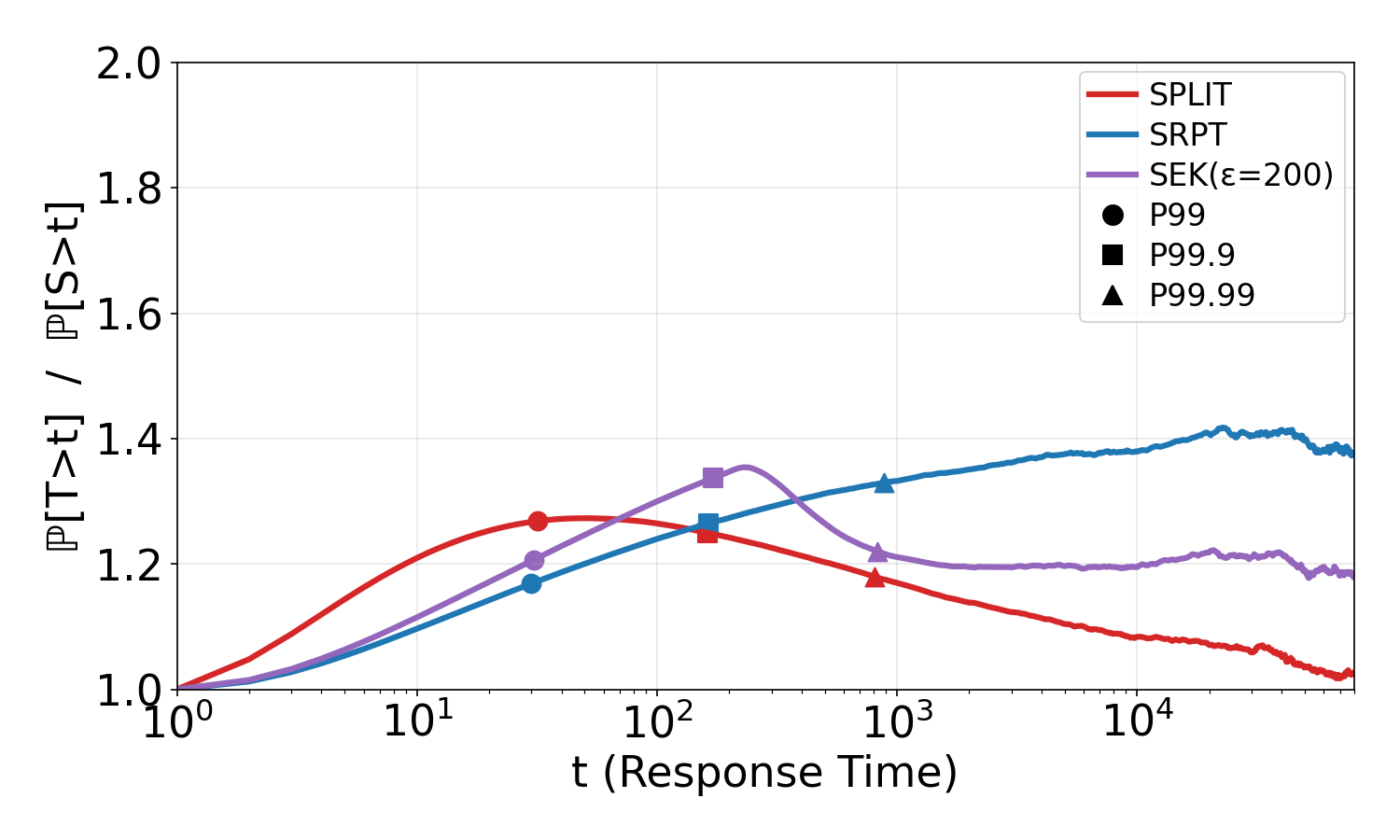}
  \caption{$\alpha=1.4$}
  \label{subfig:alpha1.4}
\end{subfigure}\hfill
\begin{subfigure}[b]{0.45\textwidth}
  \includegraphics[width=\textwidth]{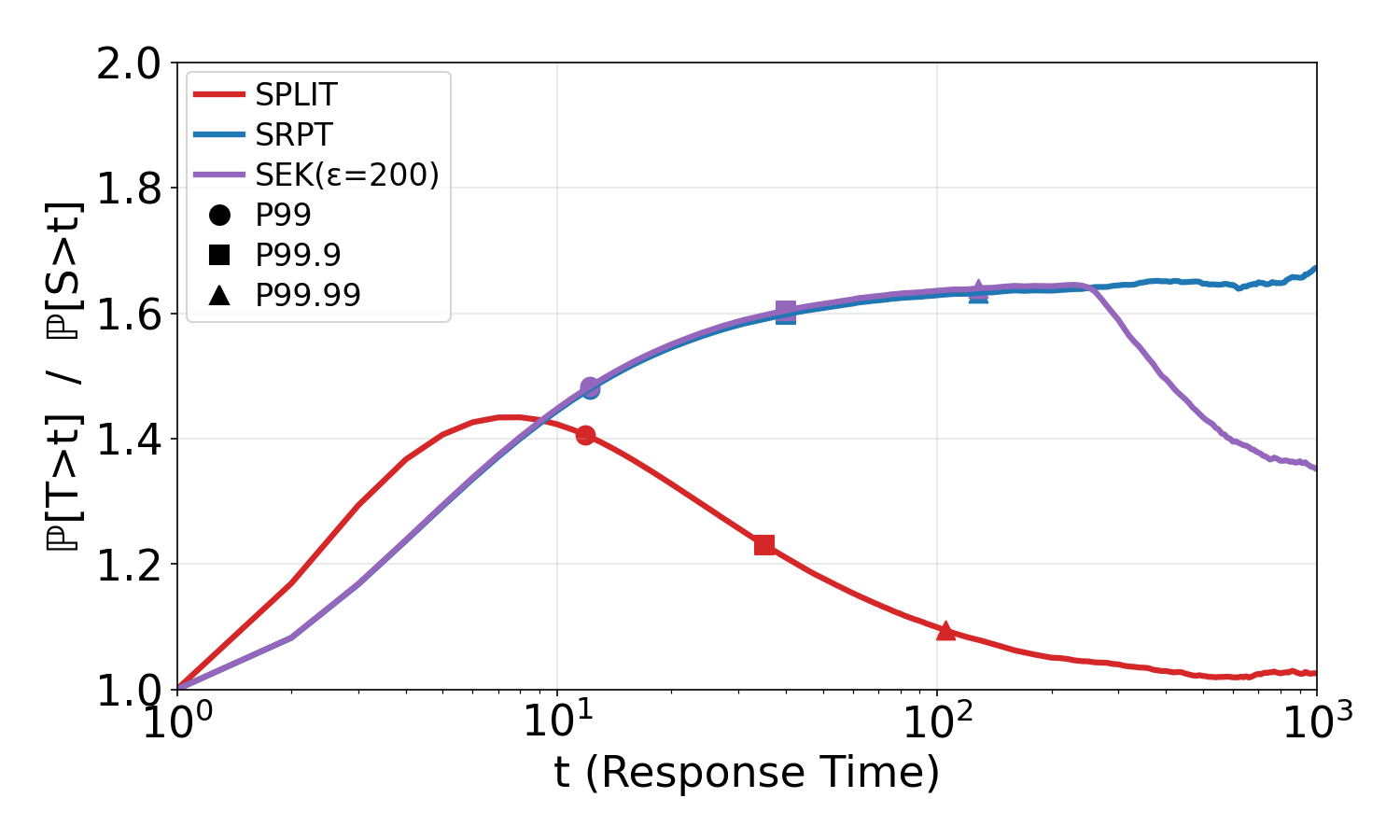}
  \caption{$\alpha=2.0$}
  \label{subfig:alpha2}
\end{subfigure}
\caption{Normalized response time tail for varying Pareto tail index, $n=3$ servers, $\rho=0.5$. }
\label{fig:sim:alpha}
\end{figure}

The qualitative observations are preserved across both values of $\alpha$: \SPLIT's normalized tail converges to $1$, the baselines remain bounded away from $1$. In addition, at $\alpha=2.0$, \SPLIT attains better $99\%$, $99.9\%$, and $99.99\%$-tile of response times than SRPT-$n$. We do not have a clean explanation for this finite-percentile gain and report it without one.
\end{addedblock}

\subsection{High load ($\rho \geq \frac{n-1}{n}$)}
\label{sec:simulation:highload}

Figure~\ref{fig:sim:highload} shows two high-load settings: $n=3$ with $\rho=0.8$, and $n=10$ with $\rho=0.94$; both use Pareto job sizes with $\alpha=1.5$. In both cases, keeping the small-job system stable requires the large-job server to carry more than $0.4$ of the service requirement. We test thresholds corresponding to large-job-server load $0.45$, $0.5$, and $0.6$. \SPLIT does not apply in this regime. 
 We observe that with properly chosen thresholds, \threshsplit{d} approaches strong tail optimality and outperforms all baselines asymptotically.

\begin{figure}[h]
\centering
\begin{subfigure}[b]{0.45\textwidth}
  \includegraphics[width=\textwidth]{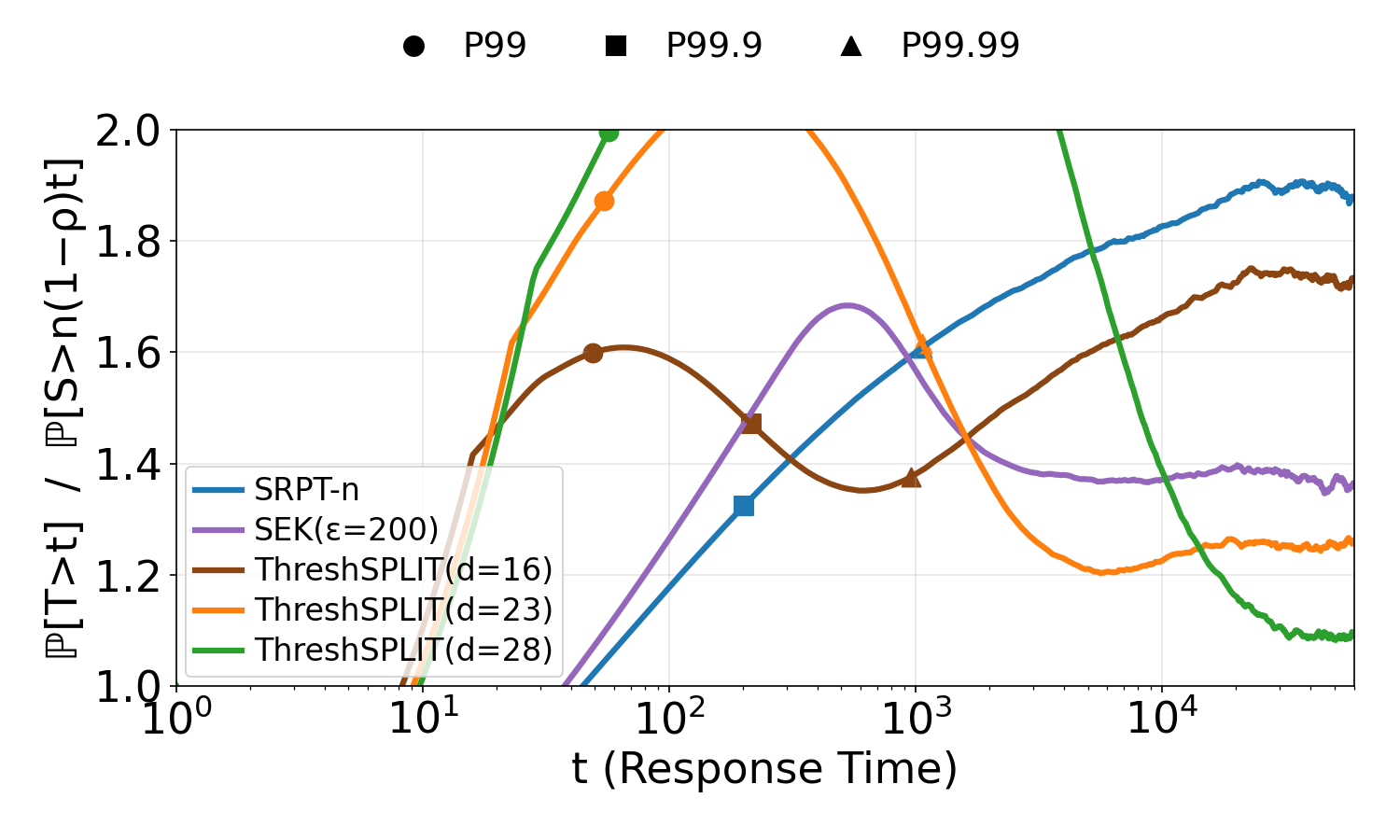}
  \caption{$n=3$, $\rho=0.8$}
  \label{subfig:highload-exp2}
\end{subfigure}\hfill
\begin{subfigure}[b]{0.45\textwidth}
  \includegraphics[width=\textwidth]{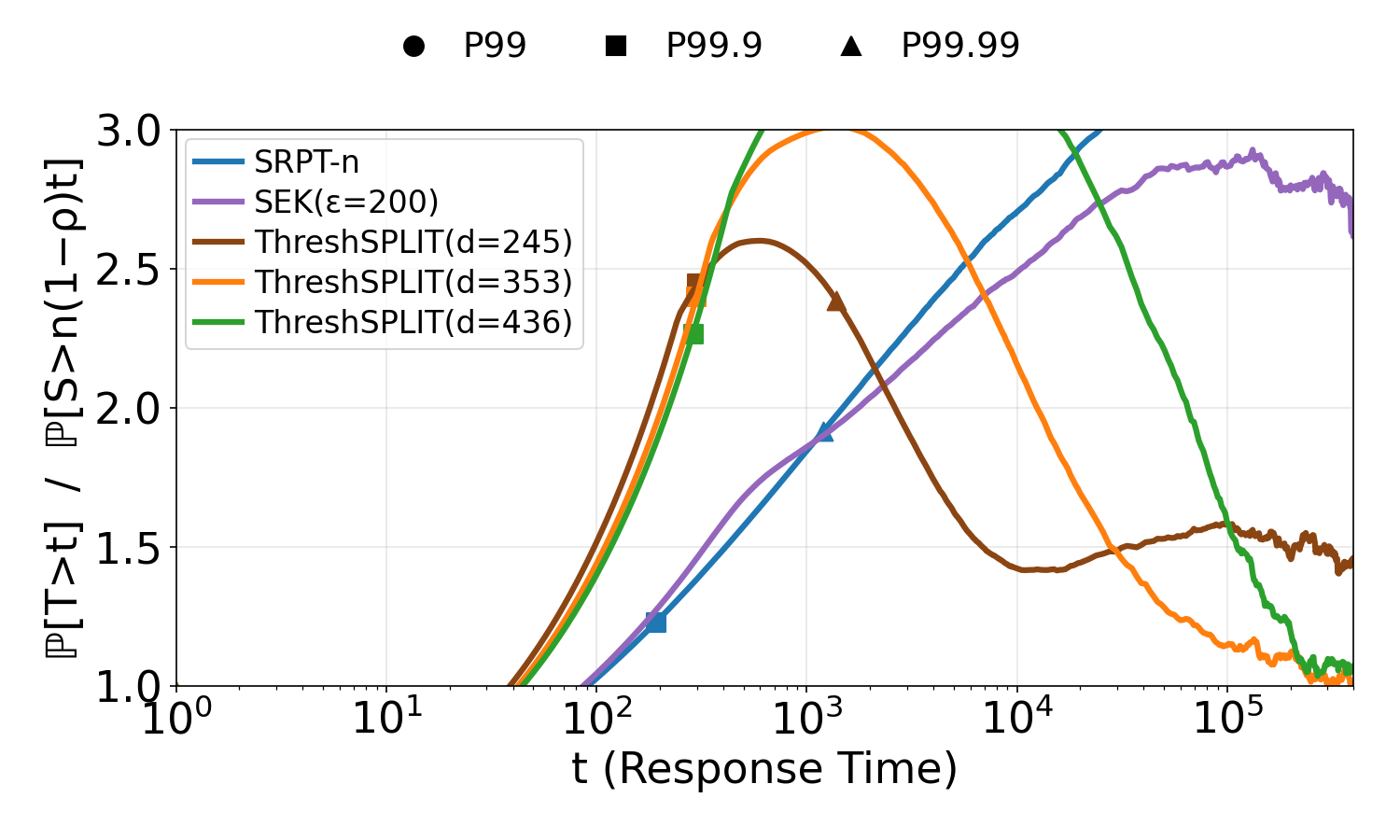}
  \caption{$n=10$, $\rho=0.94$}
  \label{subfig:highload-exp3}
\end{subfigure}
\vspace{-8pt}
\caption{Normalized response time tail in the high-load regime, Pareto job sizes with $\alpha=1.5$. \threshsplit{d} uses SRPT-$(n-1)$ on the small-job servers. }
\label{fig:sim:highload}
\end{figure}

In this regime, we use SRPT-$(n-1)$ on the small-job servers of \threshsplit{d} rather than FCFS. As discussed in Section~\ref{sec:threshsplit}, this choice is unlikely to affect the asymptotic tail: Conceivably both FCFS and SRPT-$(n-1)$ yield exponentially decaying response time tails for the bounded-size small jobs. Empirically, however, FCFS produces a noticeably larger bump in the pre-asymptotic regime at high load, which motivates the SRPT-$(n-1)$ choice here; a direct side-by-side comparison is given in Appendix~\ref{sec:appendix:fcfs-vs-srpt}.

\subsection{Performance of \TAGsplit{d}}
\label{sec:simulation:tags}

\begin{figure}[h]
\centering
\begin{subfigure}[b]{0.45\textwidth}
  \includegraphics[width=\textwidth]{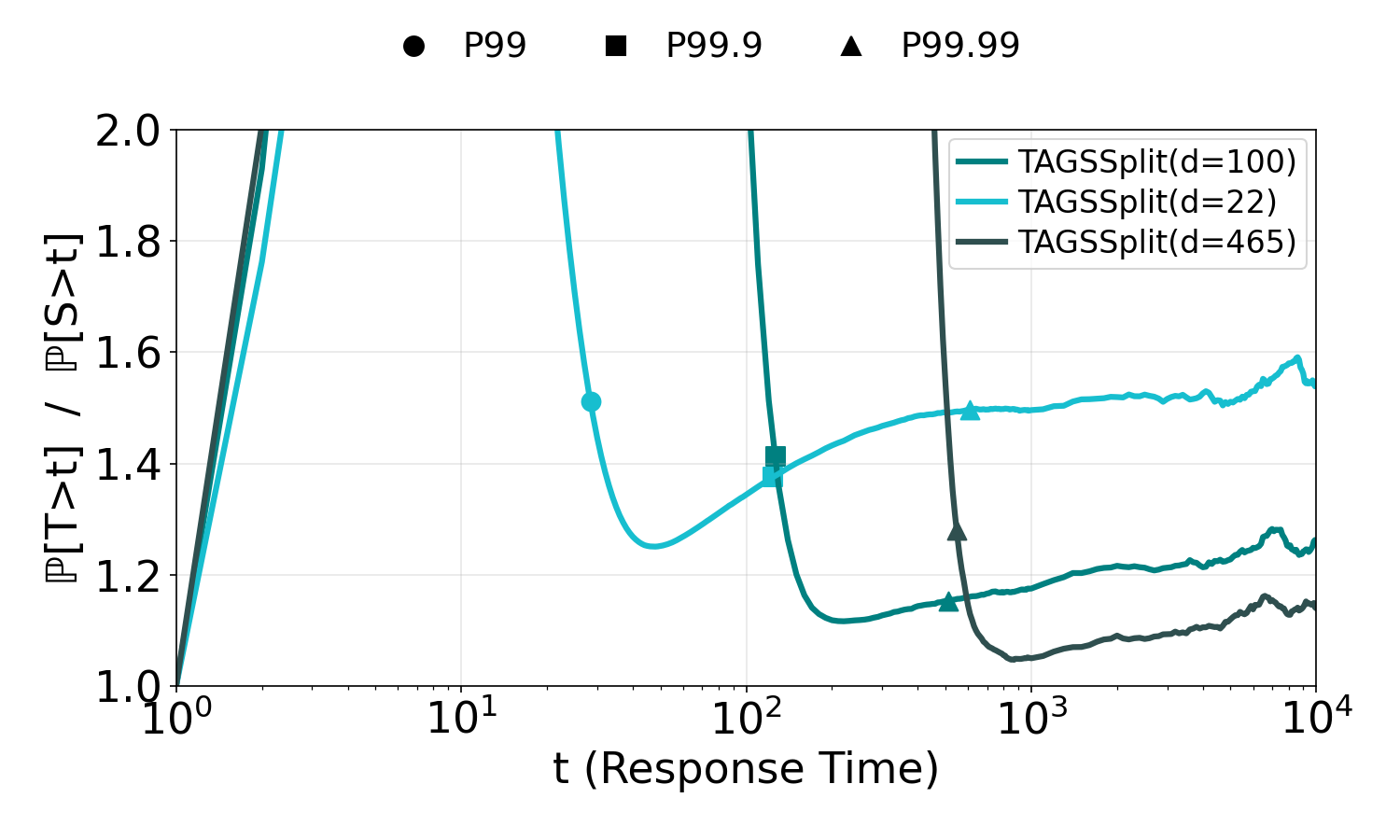}
  \caption{Low load: $n=3$, $\rho=0.5$}
  \label{subfig:tags-lowload}
\end{subfigure}\hfill
\begin{subfigure}[b]{0.45\textwidth}
  \includegraphics[width=\textwidth]{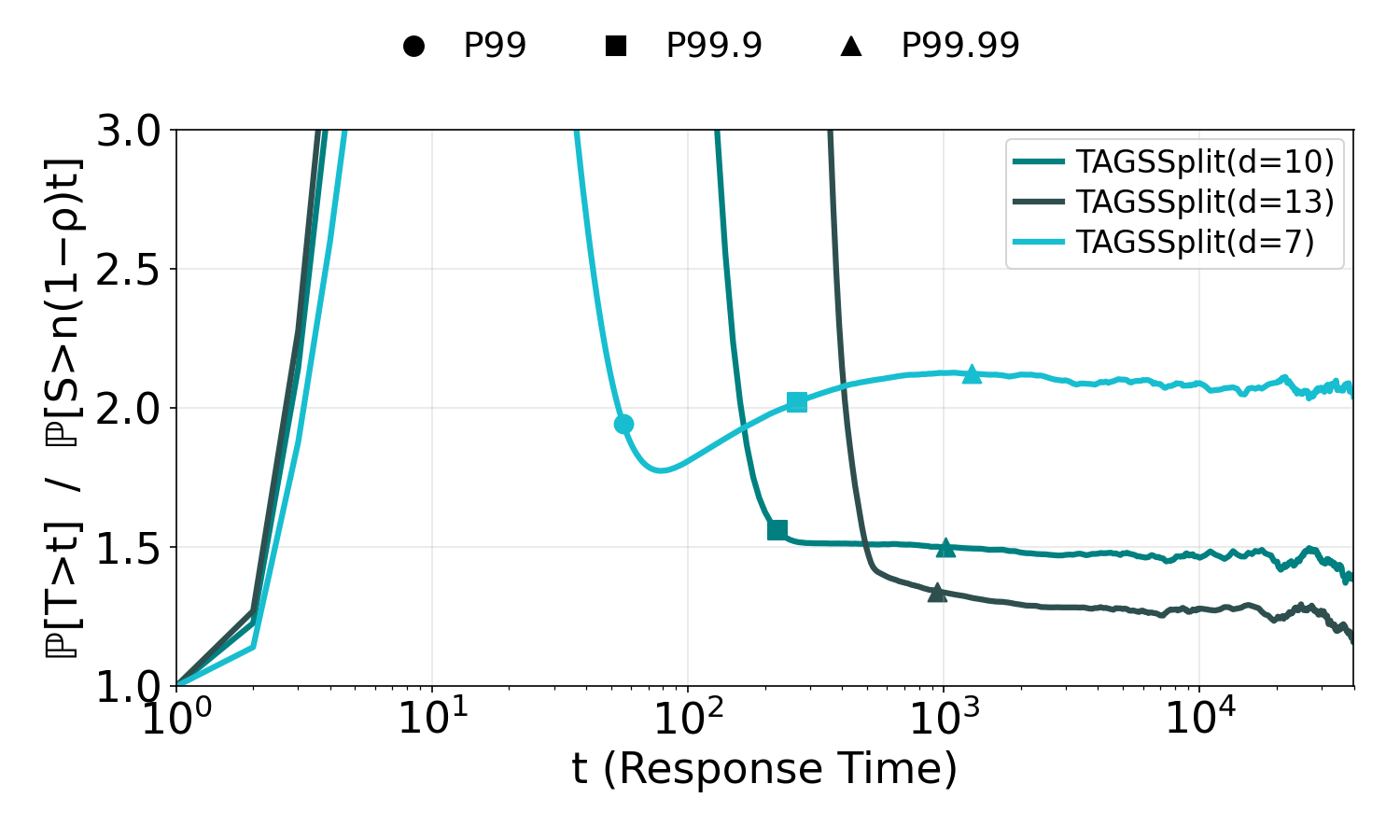}
  \caption{High load: $n=3$, $\rho=0.8$}
  \label{subfig:tags-highload}
\end{subfigure}
\vspace{-8pt}
\caption{Normalized response time tail of \TAGsplit{d}, Pareto job sizes with $\alpha=1.5$.}
\label{fig:sim:tags}
\end{figure}

Figure~\ref{fig:sim:tags} evaluates \TAGsplit{d} in the same two settings as Sections~\ref{sec:simulation:lowload} and \ref{sec:simulation:highload}. The results match Theorem~\ref{thm: TAGsplit main} in both regimes.
Figure~\ref{subfig:tags-lowload} uses the same three thresholds as \threshsplit{d} in Figure~\ref{subfig:lowload-thresh}. As $d$ grows, the normalized tail of \TAGsplit{d} approaches $1$, consistent with Theorem~\ref{thm: TAGsplit main}. Figure~\ref{subfig:tags-highload} shows the high-load setting, with thresholds chosen so that the big-job-server load $\rho^{TAGS}_{large}(d)$ is $0.45$, $0.5$, and $0.6$ (recall the definition in \eqref{eq:rho_TAG_large}). The normalized tail again approaches $1$ as $\rho^{TAGS}_{large}(d)$ decreases toward the critical value $\rho n - (n-1) = 0.4$.

%% file: 8-discussion.tex
\section{Discussion}

The policies we propose in this paper, \SPLIT, \threshsplit{d}, and \TAGsplit{d}, are all guided by two design intuitions. First, because it is the large jobs that dominate the response time tail, we dedicate separate service capacity to large jobs to reduce their response time. Second, small jobs should not be blocked by large jobs; in particular, a small job's response time should not be inflated by a catastrophically large job in the system. In this section, we reexamine these two intuitions from two perspectives: policy design (Section~\ref{sec:discussion:design}) and the effect of packing (Section~\ref{sec:discussion:packing}).

\subsection{Policy design}
\label{sec:discussion:design}

\SPLIT and \threshsplit{d} share a common structure: The $n$ servers are split into two groups, one server dedicated to large jobs and the remaining $n-1$ servers dedicated to small jobs. Despite the shared structure, \SPLIT and \threshsplit{d} schedule jobs within each group very differently: \SPLIT uses non-preemptive LJF on the large-job server and SRPT-$(n-1)$ on the small-job servers, while \threshsplit{d} uses SRPT-1 on the large-job server and FCFS on the small-job servers. We now explain the reason behind each of these choices, and distinguish the essential parts from those chosen for ease of analysis.

\paragraph{Scheduling the small jobs}
Under \SPLIT, the small-job servers must use a size-aware policy such as SRPT-$(n-1)$; FCFS would not suffice. The reason is that the LJF server handles only the single largest job in the system at a time, so whenever several big jobs are present, all but the largest remain at the small-job servers. SRPT-$(n-1)$ automatically deprioritizes these extra big jobs by their large remaining size, so small jobs continue to be served. FCFS, by contrast, would let those big jobs sit at the head of the small-job queue and potentially block subsequent small job arrivals.

Under \threshsplit{d}, this issue does not arise: Every job with size $>d$ is routed to the large-job server, so the small-job servers never see a big job. Since the small jobs' sizes are bounded by $d$, the response time tail from the small-job servers should decay exponentially under any reasonable policy, such as FCFS or SRPT-$(n-1)$ (Lemma~\ref{lemma: light tail}); our choice of FCFS is only for ease of analysis.

\paragraph{Scheduling the large jobs}
The same design intuition, that a job should not be blocked by catastrophically larger jobs, now applies one level up: A merely large job should not be blocked by a catastrophically large job either. Under \SPLIT, blocking between large jobs never occurs, because the LJF server holds only one job at a time and any merely large jobs remain at the small-job servers, where SRPT-$(n-1)$ can eventually serve them. This is why \SPLIT can afford LJF on the large-job server; the non-preemptive part is adopted only for ease of analysis.

Under \threshsplit{d}, by contrast, every job of size $>d$ is routed to a single large-job server, so the scheduling policy there must now prevent  blocking between large jobs. If \threshsplit{d} used LJF on the large-job server, every merely large job would be blocked by any catastrophically large job present, inflating its response time. Scheduling the large jobs according to SRPT avoids this. 
More generally, any single-server policy that is strongly tail optimal under heavy tails, such as PS, FB, or SMART, would work equally well on the large-job server.

\subsection{Tail benefit of \SPLIT from improved packing}
\label{sec:discussion:packing}

It is natural to ask why the \SPLIT algorithm should not also be applied within a single-server system: One could give large jobs some partial priority in a single-server queue, without delaying small jobs too much. In this subsection, we show that the tail benefit of \SPLIT (compared to SRPT-$n$) intrinsically arises from its ability to pack work across servers better. Since packing is inherently a multi-server phenomenon, no analogous benefit can be obtained in a single-server system. We develop this view in the two-server case ($n=2$) with $\rho < 1/2$ for clarity of exposition; the same intuition extends to $\rho \geq 1/2$ via \threshsplit{d}, as we discuss at the end.

\paragraph{The catastrophe principle}
Under heavy-tailed job sizes, extreme response times are driven by the \emph{catastrophe principle}: A very large response time is typically caused by a single very large job in the system, rather than by a conspiracy of many moderately large jobs \cite{nair2022fundamentals}. Understanding how a policy shapes the asymptotic response time tail therefore reduces to understanding how the policy handles a single catastrophically large job.
Concretely, consider the typical scenario that produces a large response time: The system contains one very large job of size $S^*$; all other jobs currently in the system are much smaller, as are the jobs arriving during the lifetime of this large job. A good policy must simultaneously (i) keep small jobs' response times small, so that they do not themselves contribute to the tail, and (ii) make the large job's response time, $T^*$, as small as possible.

\paragraph{Packing and the response time of the large job}
We quantify how well a policy \emph{packs} by $P_{\text{idle}}$, the fraction of the large job's service time during which the other server is idle. The total idle server-time over the large job's lifetime is therefore $P_{\text{idle}}\cdot S^*$. Over an interval of length $T^*$ after the arrival of the large job, the total work arriving to the system is approximately $2\rho\,T^*$, all of which must complete by time $T^*$ under goal~(i). Combined with the $S^*$ units of work from the large job itself and the idle server-time $P_{\text{idle}}\cdot S^*$, work conservation gives
$
2\,T^* \;\approx\; S^* \;+\; P_{\text{idle}}\cdot S^* \;+\; 2\rho\,T^*,$
which yields
\[
T^* \;\approx\; \frac{(1+P_{\text{idle}})\,S^*}{2(1-\rho)}.
\]
Hence, to reduce the asymptotic tail, one wants $P_{\text{idle}}$ as small as possible.

\paragraph{Comparing SRPT-$2$ and \SPLIT}
Under SRPT-$2$, the small jobs behave approximately as fluid: as long as at least two small jobs are present, which is almost always the case, the large job is preempted. The large job therefore runs essentially only when it is alone in the system, giving $P_{\text{idle}} \approx 1$, the worst possible packing.

Under \SPLIT, by contrast, the large job is routed to the LJF server and runs there continuously throughout its lifetime. During this time, the other server handles the stream of small jobs as a single-server queue with load $2\rho$; thus busy $2\rho$ fraction of time. Hence $P_{\text{idle}} \approx 1-2\rho$, the best packing any work-conserving policy can achieve, which is why \SPLIT attains strong tail optimality.

\paragraph{Extension to $\rho \geq 1/2$} 
The same formula for $T^*$ continues to apply when $\rho \geq 1/2$, and in this regime the goal is to drive $P_{\text{idle}}$ close to $0$. This is precisely what \threshsplit{d} accomplishes by deliberately choosing $d$ such that the load on the small-job server is $\approx 1$, keeping it busy almost all of the time during the lifetime of the catastrophically large job. 


%% file: 9-conclusion.tex
\section{Conclusion}
\label{sec:conclusion}

This paper develops the first strongly tail-optimal scheduling policies for the M/G/$n$ queue with heavy-tailed job sizes. Our central finding is that tail optimality in multi-server systems calls for a qualitatively different principle than in single-server systems: Instead of strictly prioritizing short jobs, one should give large jobs a small amount of ``sympathy'' by dedicating a server to them. This sympathy improves the asymptotic tail by enabling the system to pack work across servers more efficiently, a benefit that has no single-server analog (see Section~\ref{sec:discussion:packing}). Our three policies \SPLIT, \threshsplit{d}, and \TAGsplit{d} realize this principle across the full stability region, with and without knowledge of job sizes.

Several natural directions remain open. The primary open problem is whether exact strong tail optimality can be achieve under high load (i.e., $\rho\geq \frac{n-1}{n}$). Recall that while \threshsplit{d} can be arbitrarily close to strongly tail optimal, it never achieves exact strong tail optimality. Second, tail-optimal scheduling for multi-server systems with {\em light-tailed} job sizes is only partially understood: Boost~\cite{yu2025tale} is strongly tail optimal in heavy traffic, but nothing is known outside of heavy traffic. 
Third, it is interesting to ask whether a policy can be designed that optimizes both the mean and tail of response time.   \removed{Third, while we proved that SRPT-$n$ is weakly tail optimal for heavy-tailed job sizes (Appendix~\ref{sec:SRPT-n weak tail optimal new}), our experiments (Figure~\ref{subfig:lowload-main}) strongly suggest that SRPT-$n$ is \emph{not} strongly tail optimal; we conjecture this to be the case, and leave a rigorous proof as an open problem.}

%% file: Appendix.tex







\section{SRPT-$n$ is weakly tail optimal}
\label{sec:SRPT-n weak tail optimal new}


In this section, we prove that SRPT-$n$ is weakly tail optimal. The proof of the main lemma (Lemma \ref{lemma: SRPT-n response time bounded new}) is very similar to the proof in Section \ref{sec:SPLIT proof small jobs}. The proof of the main theorem (Theorem \ref{theorem: SRPT-n is weakly tail optimal new}) follows a similar structure in Section \ref{sec:split:main}. We state the whole proof for completeness.

Our main goal is to prove Theorem \ref{theorem: SRPT-n is weakly tail optimal new} below:
\begin{theorem}
    \label{theorem: SRPT-n is weakly tail optimal new}
    SRPT-$n$ is weakly tail optimal.
\end{theorem}

Define the response time of a job with size $x$ in an $n$-server system under SRPT-$n$ to be $T_x^{SRPT\text{-}n}$.
The main lemma is stated below:

\begin{lemma}
    \label{lemma: SRPT-n response time bounded new}
    For an $n$-server system under SRPT-$n$ with Poisson arrivals and total load $\rho<1$, for any $\beta > 0$, there exists a constant $C>0$ such that
    \begin{equation}
        \P{T_x^{SRPT\text{-}n}> Cx}= o (x^{-\beta}), \quad x \to \infty.
    \end{equation}
\end{lemma}

The proof of Lemma~\ref{lemma: SRPT-n response time bounded new} proceeds in three steps, mirroring the proof structure in Section~\ref{sec:SPLIT proof small jobs}. First, we use a potential function argument to bound the response time of a job under SRPT-$n$ in terms of the relevant work it sees upon arrival and the new small-job work that arrives during its response time (Lemma~\ref{lemma: SRPT-n response time bounded by work}). Second, we relate the relevant work under SRPT-$n$ to the relevant work in a single-server system with server speed $n$ under SRPT-1, using a known result (\removed{Lemma 2.2}\added{Lemma~2} in \cite{grosof2019srpt}). Third, we invoke Lemma~\ref{lemma: single server SRPT-1 bound} to conclude.

We first define the following notation for work.
\begin{definition}[Relevant work under SRPT-$n$]
    \label{def:relevant work SRPT-n}
    Consider a tagged job with size $x$ in an $n$-server system under SRPT-$n$. Define the relevant work of the tagged job, $W_x^{SRPT\text{-}n}$, to be the total work comprising jobs with remaining size smaller than $x$.
\end{definition}

\begin{definition}[$\WA{x}{t}$]
    Define $\WA{x}{t}$ to be the total amount of work comprising of jobs arriving after time 0 and before time $t$ with size smaller than $x$.
\end{definition}

We can now bound the response time of a job by the above two terms of work.
\begin{lemma}
    \label{lemma: SRPT-n response time bounded by work}
    Without loss of generality, assume the tagged job arrives at time 0.
    For the tagged job with size $x$, we have that
    \begin{equation}
        \P{T_x^{SRPT\text{-}n} > t}\leq \P{x + \frac{1}{n} W_x^{SRPT\text{-}n} + \frac{1}{n}\WA{x}{t}>t}.
        \label{eq:response time bounded by work SRPT-n}
    \end{equation}
\end{lemma}
\begin{proof}
    By PASTA (Poisson Arrivals See Time Averages), when the tagged job arrives, it sees $W_x^{SRPT\text{-}n}$ relevant work in the system. Under SRPT-$n$, the total work served before the tagged job is completed is upper bounded by: (1) the initial relevant work $W_x^{SRPT\text{-}n}$, (2) the $x$ work of the tagged job itself, and (3) the total work from jobs with size smaller than $x$ arriving before the tagged job is completed. The relevant work, together with newly arriving work comprising jobs with size smaller than $x$, must be completed with rate at least $n$, unless there are fewer than $n$ jobs with remaining size smaller than the tagged job and the tagged job is served. The response time of the tagged job can be upper bounded by the above argument.

    Specifically, define a potential function $G(s)$ at time $s$ to be the sum of the following three terms: (1) remaining size of the tagged job, (2) $\frac{1}{n}$ times the remaining work from the initial relevant work $W_x^{SRPT\text{-}n}$, and (3) $\frac{1}{n}$ times the remaining work from the newly arrived work $\WA{x}{s}$.

    Hence, we have that $G(0)= x + \frac{1}{n} W_x^{SRPT\text{-}n}$.
    The dynamics of $G(s)$ are given by:
    At arrivals, $G(s)$ increases by $\frac{1}{n}$ times the size of the arriving job if the arriving job is smaller than $x$. \removed{Otherwise}\added{Between arrivals}, $G(s)$ decreases with rate at least 1 \removed{if $G(s) > 0$}\added{as long as the tagged job is in the system} (if the tagged job is not served, then all $n$ servers are serving jobs with size smaller than $x$, each reducing the potential function with rate $\frac{1}{n}$; otherwise if the tagged job is served, then the first term of the potential function already decreases with rate 1).

    Mathematically, \added{integrating the decrease and adding the increases at arrivals,} we have that \added{for any $s$ before the tagged job completes,}
    \begin{removedblock}
    \begin{equation*}
        G(0) = x + \frac{1}{n} W_x^{SRPT\text{-}n},\qquad \frac{dG(s)}{ds} \leq - 1 +  \frac{1}{n}\lambda \P{S\leq x} \E{S\mid S\leq x}\quad \text{if $G(s)>0$}.
    \end{equation*}
    \end{removedblock}
    \begin{addedblock}
    \begin{equation}
        G(s) \leq x + \frac{1}{n} W_x^{SRPT\text{-}n} + \frac{1}{n}\WA{x}{s} - s.
        \label{eq:potential G(0) SRPT-n}
    \end{equation}
    \end{addedblock}

    Consider the event that the tagged job is not completed before time $t$ (which means its response time is more than $t$). In this case, \removed{the potential function $G(s)>0$ for any $s\in(0,t]$, which means the dynamics of $G(s)$ given by \eqref{eq:potential G(0) SRPT-n} always hold for any $s\in(0,t]$}\added{$G(t)>0$ since the tagged job still has positive remaining size, and \eqref{eq:potential G(0) SRPT-n} holds for $s=t$}. Hence, we have that
    \begin{equation}
        0< G(t) \leq x + \frac{1}{n} W_x^{SRPT\text{-}n} + \frac{1}{n}\WA{x}{t} - t.
    \end{equation}

    Hence we have that
    \[\P{T_x^{SRPT\text{-}n} > t}\leq \P{x + \frac{1}{n} W_x^{SRPT\text{-}n} + \frac{1}{n}\WA{x}{t}>t}. \]
\end{proof}

Now we relate the relevant work under SRPT-$n$ to the relevant work in a single-server system under SRPT-1.

\begin{lemma}[\removed{Lemma 2.2}\added{Lemma~2} in \cite{grosof2019srpt}]
    \label{lemma: W_x SRPT-n bounded}
    Given an $n$-server system under SRPT-$n$ with Poisson arrivals and total load $\rho<1$. Consider a single-server system under SRPT-1 with the same arrival process but the server has speed $n$. Define $\Wsuper{x}$ to be the relevant work that a job with size $x$ sees in this single-server system. Then we have that
    \begin{equation}
        W_x^{SRPT\text{-}n} \leq \Wsuper{x} + nx.
    \end{equation}
\end{lemma}

Note that although it is assumed that the job size distribution has finite variance in \cite{grosof2019srpt}, the proof for this lemma does not require this assumption. \added{Moreover, Lemma~2 in \cite{grosof2019srpt} is stated for $n$ servers of speed $\frac{1}{n}$ each and a single server of speed $1$; since it holds for any arrival sequence, rescaling time by a factor of $n$ gives Lemma~\ref{lemma: W_x SRPT-n bounded}.}

For simplicity, now consider a single-server system with server speed 1 but all job sizes (and consequently the work) are scaled by $\frac{1}{n}$. Define the scaled notation $\Wstar{x}:=\frac{1}{n}\Wsuper{x}$ and $\WAstar{x}{t}:=\frac{1}{n} \WA{x}{t}$. This allows us to use the known result for single-server systems under SRPT-1, see Lemma~\ref{lemma: single server SRPT-1 bound}.

We are now ready to prove Lemma~\ref{lemma: SRPT-n response time bounded new}.

\begin{proof}[Proof of Lemma~\ref{lemma: SRPT-n response time bounded new}]
    Let $C$ be the constant from Lemma~\ref{lemma: single server SRPT-1 bound} and consider a tagged job with size $x$. By Lemma~\ref{lemma: SRPT-n response time bounded by work} with $t = Cx$,
    \begin{equation}
        \P{T_x^{SRPT\text{-}n} > Cx} \leq \P{x + \frac{1}{n} W_x^{SRPT\text{-}n} + \frac{1}{n}\WA{x}{Cx} > Cx}.
    \end{equation}
    By Lemma~\ref{lemma: W_x SRPT-n bounded}, $W_x^{SRPT\text{-}n} \leq \Wsuper{x} + nx$. Substituting and using $\Wstar{x} = \frac{1}{n}\Wsuper{x}$ and $\WAstar{x}{Cx} = \frac{1}{n}\WA{x}{Cx}$, we obtain
    \begin{align*}
        &\P{x + \frac{1}{n} W_x^{SRPT\text{-}n} + \frac{1}{n}\WA{x}{Cx} > Cx} \\
        &\leq \P{2x + \Wstar{x} + \WAstar{x}{Cx} > Cx} \\
        &= o(x^{-\beta}). && \text{Lemma~\ref{lemma: single server SRPT-1 bound}}
    \end{align*}
\end{proof}

We can now prove Theorem~\ref{theorem: SRPT-n is weakly tail optimal new}.

\begin{proof}[Proof of Theorem~\ref{theorem: SRPT-n is weakly tail optimal new}]

Pick $\beta>\alpha$ and let $C$ be the constant from Lemma~\ref{lemma: SRPT-n response time bounded new}. Then we have that
\begin{align*}
    &\lim_{t\to \infty} \frac{\P{T^{SRPT\text{-}n}>t}}{\P{S>t}} \\
    &= \lim_{t\to \infty} \frac{1}{\P{S>t}} \int_{0}^\infty \P{T_x^{SRPT\text{-}n}>t} \cdot f_S(x) dx \\
    &= \lim_{t\to \infty} \frac{1}{\P{S>t}} \left(\P{S>t} + \int_{0}^{t/C} \P{T_x^{SRPT\text{-}n}>t} \cdot f_S(x) dx +\int_{t/C}^t \P{T_x^{SRPT\text{-}n}>t} f_S(x) dx \right)\\
    &\leq 1 + \lim_{t\to \infty} \frac{1}{\P{S>t}}\left( \int_{0}^{t/C} \P{T_{t/C}^{SRPT\text{-}n}> t} \cdot f_S(x) dx +\int_{t/C}^t 1\cdot f_S(x) dx  \right)\\
    &\leq 1 + \lim_{t\to\infty} \frac{\P{T_{t/C}^{SRPT\text{-}n}> t}}{\P{S>t}} \cdot \P{S\leq t/C} + \lim_{t\to\infty} \frac{1}{\P{S>t}} \cdot \P{S>t/C}\\
    &\leq 1 + \lim_{t\to\infty} \frac{\P{T_{t/C}^{SRPT\text{-}n}> t}}{\P{S> t/C}} \cdot \frac{\P{S> t/C}}{\P{S>t}}  \cdot 1 + \lim_{t\to\infty} \frac{1}{\P{S>t}} \cdot \P{S>t/C}\\
    &= 1 + 0 + C^\alpha.
\end{align*}
Here the last equation uses Lemma~\ref{lemma: SRPT-n response time bounded new}.
\end{proof}

\begin{addedblock}
\section{SRPT-$n$ is not strongly tail optimal}
\label{sec:SRPT-n not strongly tail optimal}

In this section, we prove that SRPT-$n$ is not strongly tail optimal for any $n\geq 2$ and any load $\rho<1$ (Theorem~\ref{thm: SRPT-n not strongly tail optimal}), complementing Theorem~\ref{theorem: SRPT-n is weakly tail optimal new}. We first show the intuitions and then the formal proofs.

\paragraph{Intuitive arguments:}
Roughly speaking, Lower bound 1 and 2 in Section~\ref{sec:theory:lowerbound} say that no policy can complete a very large job faster than rate $\min\{1,n(1-\rho)\}$, and which of the two terms is smaller depends on the load. When $\rho<\frac{n-1}{n}$, we have $n(1-\rho)>1$, and Lower bound~2 is tighter; when $\rho\geq\frac{n-1}{n}$, we have $n(1-\rho)\leq 1$, and Lower bound~1 is tighter, which means, roughly, that the very large jobs need to run with rate $n(1-\rho)$ on average.

Under SRPT-$n$, a very large job has lower priority than all other jobs for almost all of its service, so it is served only when the other jobs leave at least one server idle, and then on a single server. Let $N$ denote the steady-state number of jobs in an SRPT-$n$ system, and define
\[
    q := \P{N<n},
\]
the fraction of time during which at least one server is idle. The very large job is then completed at rate roughly $q$, which we will show is strictly smaller than $\min\{1,n(1-\rho)\}$ under both loads. 

When $\rho<\frac{n-1}{n}$, the relevant comparison is $q<1$: Completing the very large job at rate $1$ would require serving it all the time, but under SRPT-$n$ it waits whenever the other jobs occupy all $n$ servers, which happens a positive fraction of the time. When $\rho\geq\frac{n-1}{n}$, the relevant comparison is $q<n(1-\rho)$: Here $n(1-\rho)$ is the average number of idle servers, whereas the very large job uses at most one idle server at a time. Whenever two or more servers are idle at the same time, the capacity of all but one of them is lost, and this lost capacity is the reason why SRPT-$n$ is not strongly tail optimal.

\paragraph{Formal arguments:}

\begin{theorem}
    \label{thm: SRPT-n not strongly tail optimal}
    Suppose $n\geq 2$ and $\rho<1$. Then
    \begin{equation}
        \liminf_{t\to\infty}\frac{\P{T^{SRPT\text{-}n}>t}}{\P{S>\min\{1,n(1-\rho)\}\, t}} \geq \left(\frac{\min\{1,n(1-\rho)\}}{q}\right)^{\alpha}>1.
        \label{eq:SRPT-n strict gap}
    \end{equation}
    Consequently, SRPT-$n$ is not strongly tail optimal.
\end{theorem}

The proof has three steps. First, Lemma~\ref{lemma: SRPT-n idle probability} shows that $q<\min\{1,n(1-\rho)\}$, which gives the strict inequality in \eqref{eq:SRPT-n strict gap}. Second, Lemma~\ref{lemma: SRPT-n large job slowdown} shows that a very large job is indeed completed at rate at most roughly $q$, which gives the first inequality in \eqref{eq:SRPT-n strict gap}. Third, to conclude that SRPT-$n$ is not strongly tail optimal, we compare it with \threshsplit{d}, whose tail gets arbitrarily close to $\P{S>\min\{1,n(1-\rho)\}t}$ by Theorem~\ref{thm: main}.

\begin{lemma}
    \label{lemma: SRPT-n idle probability}
    Suppose $n\geq 2$ and $\rho<1$. Then $0<q<\min\{1,n(1-\rho)\}$.
\end{lemma}
\begin{proof}
    Under SRPT-$n$, exactly $\min\{N,n\}$ servers are busy, and in steady state the rate at which work is completed equals the rate $\lambda\E{S}=n\rho$ at which work arrives. Hence $\E{\min\{N,n\}}=n\rho$, and the average number of idle servers is
    \[
        n(1-\rho) = \E{(n-N)^+} = \P{N<n} + \E{(n-N-1)^+} = q + \E{(n-N-1)^+},
    \] 
    where $z^+:=\max\{z,0\}$, and the second equality holds because $(n-N)^+ = \mathbf{1}\{N<n\}+(n-N-1)^+$. A stable M/G/$n$ system is empty with positive probability, so $\E{(n-N-1)^+}\geq (n-1)\P{N=0}>0$, where we use $n\geq 2$. Hence $q<n(1-\rho)$, and $q\geq \P{N=0}>0$. Finally, $\P{N\geq n}>0$, since with positive probability $n$ jobs arrive to an empty system within $\epsilon$ time units and each of them has size larger than $\epsilon$, for any $\epsilon>0$ with $\P{S>\epsilon}>0$. Hence $q<1$.
\end{proof}

Next, we show that a very large job is roughly completed at rate at most $q$. The difficulty is that 
at some point, the very large job will get high priority: Once its remaining size is small, it overtakes other jobs. The proof of the next lemma therefore only tracks the service needed to reduce its remaining size to a small fraction $\delta$ of its size.

\begin{lemma}
    \label{lemma: SRPT-n large job slowdown}
    Suppose $n\geq 2$ and $\rho<1$. Then for any constant $c>q$,
    \[
        \lim_{x\to\infty}\P{T_x^{SRPT\text{-}n}>\frac{x}{c}} = 1.
    \]
\end{lemma}
\begin{proof}
    Consider a tagged job of size $x$ that arrives at time $0$, and consider an auxiliary SRPT-$n$ system with the same jobs present at time $0^-$ and the same arrivals after time $0$, but without the tagged job. By PASTA, as in the proof of Lemma~\ref{lemma: response time bounded by work}, the auxiliary system is in steady state at time $0$, and its state does not depend on $x$. Let $N(s)$ be the number of jobs in the auxiliary system at time $s$, and let
    \[
        A(u) := \int_0^u \mathbf{1}\{N(s)<n\}\, ds
    \]
    be the total time during $[0,u]$ at which the auxiliary system has at least one idle server. By the ergodic theorem, $\frac{A(u)}{u}\to q$ almost surely as $u\to\infty$.

    Fix $c>q$ and choose $\delta\in(0,1)$ such that $\frac{q}{c}<1-\delta$. Consider the following three events:
    \begin{enumerate}
        \item Every job present at time $0^-$ has remaining size at most $\delta x$.
        \item Every job arriving during $(0,\frac{x}{c}]$ has size at most $\delta x$.
        \item $A(\frac{x}{c})<(1-\delta)x$.
    \end{enumerate}
    We first show that if all three events occur, then $T_x^{SRPT\text{-}n}>\frac{x}{c}$. By events 1 and 2, every job other than the tagged job has remaining size at most $\delta x$ during $[0,\frac{x}{c}]$. Hence, as long as the remaining size of the tagged job is larger than $\delta x$, every other job has higher priority than the tagged job. Since SRPT-$n$ serves the $n$ jobs with the smallest remaining sizes, the tagged job then only occupies a server that is idle in the auxiliary system: The other jobs receive exactly the same service as in the auxiliary system, and the tagged job is served, on a single server, exactly when $N(s)<n$. To reduce its remaining size from $x$ to $\delta x$, the tagged job needs $(1-\delta)x$ units of service, but by event 3, it receives less than $(1-\delta)x$ units of service by time $\frac{x}{c}$. Hence, at time $\frac{x}{c}$, its remaining size is still larger than $\delta x$, and in particular it has not completed.

    It remains to show that each of the three events occurs with probability tending to $1$ as $x\to\infty$. First, the largest remaining size among the jobs present at time $0^-$ is finite almost surely, and its distribution does not depend on $x$, so event 1 occurs with probability tending to $1$. Second, jobs with size larger than $\delta x$ arrive according to a Poisson process with rate $\lambda\P{S>\delta x}$, so
    \[
        \P{\text{event 2}} = \exp\left(-\lambda\,\frac{x}{c}\,\P{S>\delta x}\right)\to 1,
    \]
    where we use $\delta x\,\P{S>\delta x}\leq \E{S\mathbf{1}\{S>\delta x\}}\to 0$, since $\E{S}<\infty$. 
    
    Third, almost surely, \[\frac{A(x/c)}{x} = \frac{1}{c}\cdot\frac{A(x/c)}{x/c}\to\frac{q}{c}<1-\delta.\]  So event 3 occurs with probability tending to $1$. By a union bound over the complements of the three events, all three events occur with probability tending to $1$, which proves the lemma.
\end{proof}

We are now ready to prove Theorem~\ref{thm: SRPT-n not strongly tail optimal}.

\begin{proof}[Proof of Theorem~\ref{thm: SRPT-n not strongly tail optimal}]
    Fix a constant $c>q$. For a tagged job of size $x>ct$, we have $\frac{x}{c}>t$, so $\P{T_x^{SRPT\text{-}n}>t}\geq \P{T_x^{SRPT\text{-}n}>\frac{x}{c}}$. Hence
    \[
        \P{T^{SRPT\text{-}n}>t} \geq \P{T^{SRPT\text{-}n}>t,\ S>ct} \geq \left(\inf_{x>ct}\P{T_x^{SRPT\text{-}n}>\frac{x}{c}}\right)\P{S>ct},
    \]
    where the infimum tends to $1$ as $t\to\infty$ by Lemma~\ref{lemma: SRPT-n large job slowdown}. Dividing by $\P{S>\min\{1,n(1-\rho)\}t}$ and using \eqref{eq:regular varying definition}, we have that
    \[
        \liminf_{t\to\infty}\frac{\P{T^{SRPT\text{-}n}>t}}{\P{S>\min\{1,n(1-\rho)\}t}} \geq \lim_{t\to\infty}\frac{\P{S>ct}}{\P{S>\min\{1,n(1-\rho)\}t}} = \left(\frac{\min\{1,n(1-\rho)\}}{c}\right)^{\alpha}.
    \]
    Letting $c\downarrow q$ gives the first inequality in \eqref{eq:SRPT-n strict gap}, and Lemma~\ref{lemma: SRPT-n idle probability} gives the second.

    It remains to exhibit a policy whose tail is strictly lighter than that of SRPT-$n$. By the proof of Theorem~\ref{thm: main}, as the threshold $d$ increases to $\infty$ (if $\rho\leq\frac{n-1}{n}$) or to $d^*$ (if $\rho>\frac{n-1}{n}$), the small-job system of \threshsplit{d} remains stable and $1-r_{>d}$ increases to $\min\{1,n(1-\rho)\}$. Since $\min\{1,n(1-\rho)\}>q$, we can choose such a threshold $d$ with $1-r_{>d}>q$. By Theorem~\ref{thm: split tail}, \eqref{eq:SRPT-n strict gap}, and \eqref{eq:regular varying definition},
    \[
        \liminf_{t\to\infty}\frac{\P{T^{SRPT\text{-}n}>t}}{\P{T^{\threshsplit{d}}>t}} = \liminf_{t\to\infty}\frac{\P{T^{SRPT\text{-}n}>t}}{\P{S>(1-r_{>d})t}} \geq \left(\frac{1-r_{>d}}{q}\right)^{\alpha}>1.
    \]
    Hence, the supremum in the definition of tail optimality (Section~\ref{sec:intro:tail}) is strictly larger than $1$ for SRPT-$n$, so SRPT-$n$ is not strongly tail optimal.
\end{proof}
\end{addedblock}

\section{Results for single-server systems under SRPT-1}

The goal of this appendix is to prove the following lemma for single-server systems under SRPT-1 policy:

\begin{lemma}
    \label{lemma: single server SRPT-1 bound}
    For a single-server system under SRPT-1 policy with load $\rho<1$, we have that for any $\beta>0$, there exists a constant $C>0$ such that
    \begin{equation}
        \P{2x + \Wstar{x} + \WAstar{x}{Cx} > Cx} = o(x^{-\beta}).
    \end{equation}
\end{lemma}

The proof follows a conventional method for single-server systems under SRPT-1 (e.g., see the proof of Lemma 18 in \cite{nuyens2008preventing}). It uses the following lemma, a standard result for random walks (see, e.g., Lemma 16 in \cite{nuyens2008preventing}, which is due to \cite{resnick1999activity}).
\begin{lemma}[Lemma 16 in \cite{nuyens2008preventing}]
    \label{lemma:random walk tail}
    Let $X_i$ be a sequence of i.i.d. random variables with $\E{X_i} < 0$ and $\E{(X_i^+)^p} < \infty$ for some $p > 1$. Let $\mathcal{S}_n(x):=\sum_{i=1}^{n} (\min\{X_i,x\})$ for some $x>0$. Define $M(x):=\sup_{n} \mathcal{S}_n(x)$. Then, for any $0<\beta < \infty$, there exist constants $k > 0$ such that
    \[
        \P{ M(x) > k x} = o(x^{-\beta}).
    \]
\end{lemma}

\begin{proof}[Proof of Lemma \ref{lemma: single server SRPT-1 bound}]
    Define $S_i$ to be the size of the $i$th arrival and $A_i$ to be the $i$th interarrival time. Now define an auxiliary sample path with the same interarrival times ($A_i':=A_i$) but $S_i':=\min(S_i,x)$. Define $W_A'(t)$ to be the total work from jobs in the auxiliary sample path which arrive after time 0 and before time $t$. Then we have that $W_A'(t)\geq \WAstar{x}{t}$. Define $\rho'$ to be the load of the auxiliary sample path, i.e., $\rho' := \frac{\E{S_i'}}{\E{A_i'}}<1$.

    Define $W_x'$ to be the steady-state work in the auxiliary sample path. Now we prove that $W_x' + x \geq \Wstar{x}$: In the original sample path, denote the consecutive period of time where $\Wstar{x}>0$ by ``$x$ busy period''. Then, during each $x$ busy period, $\Wstar{x}$ only increases when a job with size smaller than $x$ arrives, while $W_x'$ increases when every job arrives (if the job size is larger than $x$, $W_x'$ increases by $x$). Hence, we have that the difference $W_x'- \Wstar{x}$ never decreases during an $x$ busy period. Moreover, at the start of each $x$ busy period, $\Wstar{x}\leq x$. This gives the proof of the inequality $W_x' + x \geq \Wstar{x}$.

    Define the auxiliary quantity
    \begin{equation}
        U_x^c := \sup_{t > 0} [W_A'(t) - ct].
    \end{equation}
    By the Loynes construction \cite{loynes1962stability}, $U_x^1$ is the steady-state relevant work in the auxiliary sample path, i.e., $U_x^1 =_{st} W_x'$. Hence we have that $U_x^1 + x \geq_{st}  \Wstar{x}$.

    Pick $\delta < 1-\rho'$ and let $C \geq \frac{12}{\delta}$. Then we have that
    \begin{align*}
        &\P{2x + \Wstar{x} + \WAstar{x}{Cx} > Cx} \\
        &\leq \P{3x + U_x^1 + W_A'(Cx) > Cx} \\
        &= \P{U_x^1 + W_A'(Cx) - (1-\delta)Cx + 3x > \delta Cx} \\
        &\leq \P{U_x^1 > \frac{\delta Cx}{2}} + \P{W_A'(Cx) - (1-\delta)Cx + 3x > \frac{\delta Cx}{2}}.
    \end{align*}
    For the first term, note that $U_x^1 \leq U_x^{1-\delta}$, hence
    \[
        \P{U_x^1 > \frac{\delta Cx}{2}} \leq \P{U_x^{1-\delta} > \frac{\delta Cx}{2}}.
    \]
    For the second term, note that $W_A'(Cx) - (1-\delta)Cx \leq \sup_{t>0}[W_A'(t) - (1-\delta)t] = U_x^{1-\delta}$, hence we have that
    \[
        \P{W_A'(Cx) - (1-\delta)Cx + 3x > \frac{\delta Cx}{2}} \leq \P{U_x^{1-\delta} >\left(\frac{\delta C}{2}-3\right)x} \leq \P{U_x^{1-\delta} > \frac{\delta Cx}{4}}.
    \]
    Hence,
    \begin{equation}
        \P{2x + \Wstar{x} + \WAstar{x}{Cx} > Cx} \leq \P{U_x^{1-\delta} > \frac{\delta Cx}{2}} + \P{U_x^{1-\delta} > \frac{\delta Cx}{4}}\leq 2\P{U_x^{1-\delta} > \frac{\delta Cx}{4}}.
    \end{equation}

    It remains to show $\P{U_x^{1-\delta} > \frac{\delta Cx}{4}} = o(x^{-\beta})$.
    Because the supremum in $U_x^{1-\delta}$ is attained at arrival instants, we can write $U_x^{1-\delta} = \sup_n \sum_{i=1}^{n} [S_i' - (1-\delta)A_i']$. Because $1-\delta > \rho'$, the step sizes $X_i := S_i' - (1-\delta)A_i'$ have negative mean. Moreover, each $S_i'$ is bounded above by $x$, so $\min\{X_i, x\} = X_i$. Applying Lemma \ref{lemma:random walk tail} with $M(x) = U_x^{1-\delta}$, for any $\beta > 0$, there exists a constant $\bar{k} > 0$ such that
    \[
        \P{U_x^{1-\delta} > \bar{k} x} = o(x^{-\beta}).
    \]
    Setting $C = \max\{\frac{4\bar{k}}{\delta},\frac{12}{\delta}\}$ completes the proof.
\end{proof}

\section{Proof of Theorem~\ref{thm: TAGsplit tail}}
\label{appendix:proof TAGsplit tail}

\TAGsplittail*

\begin{proof}
\noindent\textbf{Small jobs.} For jobs with size $S \leq d$, their response time under \TAGsplit{d} is that of a FCFS system with $n-1$ servers and job sizes bounded by $d$. By Lemma~\ref{lemma: light tail}, there exists a constant $\gamma>0$ such that
    \begin{equation}
        \lim_{t\to \infty} \frac{\P{T^{\TAGsplit{d}}>t \mid S \leq d}}{e^{-\gamma t} }\leq 1.
        \label{eq:TAGsplit proof small jobs tail}
    \end{equation}

\noindent\textbf{Big jobs.} For jobs with size $S > d$, the response time under \TAGsplit{d} consists of three parts: (i) the waiting time in the FCFS queue, denoted by $T_{wait}^{FCFS}$, (ii) $d$ units of service at the FCFS servers, and (iii) the time at the PS server, denoted by $T_{PS}$. Part~(i) decays exponentially by Lemma~\ref{lemma: light tail}, and part~(ii) is a constant. Part~(iii) is the response time of a single-server PS system with load $\rho^{TAGS}_{large}(d)$ and job size distribution $(S-d \mid S>d)$. By Approximation~\ref{approx:TAGS} and \cite{guillemin2004tail}, the PS tail satisfies
    \begin{equation}
        \lim_{t\to \infty} \frac{\P{T_{PS}>t}}{\P{S-d>(1-\rho^{TAGS}_{large}(d))t \mid S>d}} = 1.
        \label{eq:TAGsplit proof large jobs tail}
    \end{equation}

Hence we have that the tail of the response time of a big job under \TAGsplit{d} satisfies
\begin{align}
    \lim_{t\to\infty} \frac{\P{T^{\TAGsplit{d}} > t \mid S>d}}{\P{S > (1-\rho^{TAGS}_{large}(d))t \mid S>d}}
      &= \lim_{t\to\infty} \frac{\P{T_{wait}^{FCFS} + d + T_{PS} > t}}{\P{S > (1-\rho^{TAGS}_{large}(d))t \mid S>d}} \nonumber\\
      &= \lim_{t\to\infty} \frac{\P{T_{PS} > t}}{\P{S > (1-\rho^{TAGS}_{large}(d))t \mid S>d}} \nonumber\\
      &= \lim_{t\to\infty} \frac{\P{T_{PS} > t}}{\P{S-d > (1-\rho^{TAGS}_{large}(d))t \mid S>d}} \nonumber\\
      &= 1,
    \label{eq:TAGsplit proof big job tail}
\end{align}
where the second equality uses that $T_{wait}^{FCFS}$ has exponentially decaying tail (part (i)) and $d$ is a constant (part (ii)), both negligible compared to the regularly varying $T_{PS}$;\footnote{This argument can be made rigorous trivially by splitting $t$ into $\epsilon t$ and $(1-\epsilon)t$, and then using a union bound to bound both terms.} the third uses that $\P{S-d > x \mid S>d} \sim \P{S > x \mid S>d}$ as $x\to\infty$ for regularly varying $S$; and the last equality uses \eqref{eq:TAGsplit proof large jobs tail}.

\noindent\textbf{Combining.} By the law of total probability,
    \begin{align*}
        &\lim_{t\to \infty} \frac{\P{T^{\TAGsplit{d}}>t}}{\P{S>(1-\rho^{TAGS}_{large}(d))t}} \\
        &= \lim_{t\to \infty} \frac{\P{T^{\TAGsplit{d}}>t \mid S \leq d} \P{S \leq d} + \P{T^{\TAGsplit{d}}>t \mid S>d} \P{S>d}}{\P{S>(1-\rho^{TAGS}_{large}(d))t}}.
    \end{align*}
The first term vanishes because the exponential decay in \eqref{eq:TAGsplit proof small jobs tail} is faster than any heavy-tailed distribution. The second term converges to 1 by \eqref{eq:TAGsplit proof big job tail}.
\end{proof}

\section{Additional simulation results}
\label{sec:appendix:additional-sim}

\removed{[Former Appendix D.1, ``Other Pareto tail indices'' (figure and text), moved to the end of Section~\ref{sec:simulation:lowload}.]}
%
%
%

\subsection{\threshsplit{d} simulation under different scheduling policies for the small jobs}
\label{sec:appendix:fcfs-vs-srpt}

As discussed in Sections~\ref{sec:threshsplit} and \ref{sec:discussion:design}, the choice of scheduling policy on the $n-1$ small-job servers in \threshsplit{d} is not essential: because small jobs' sizes are bounded by $d$, any reasonable policy conceivably yields an  exponentially decaying response time tail for small jobs (Lemma~\ref{lemma: light tail}), so the asymptotic tail of \threshsplit{d} matches Theorem~\ref{thm: main} under either FCFS or SRPT-$(n-1)$. In Section~\ref{sec:simulation}, the low-load figures (Section~\ref{sec:simulation:lowload}) use FCFS on the small-job servers, while the high-load figures (Section~\ref{sec:simulation:highload}) use SRPT-$(n-1)$; the reason is that FCFS produces a visibly larger bump in the pre-asymptotic regime at high load. In this appendix, we show the two variants side by side on the three simulation settings of Section~\ref{sec:simulation} to validate both points empirically.

Across all three settings, the asymptotic behavior of the two variants coincides, matching our theoretical prediction. In the pre-asymptotic regime, the two choices give similar curves at low load (Figure~\ref{fig:app:fcfs-vs-srpt-low}), while FCFS exhibits a visibly larger bump at high load (Figures~\ref{fig:app:fcfs-vs-srpt-exp2}~and~\ref{fig:app:fcfs-vs-srpt-exp3}).

\begin{figure}[h]
\centering
\begin{subfigure}[b]{0.45\textwidth}
  \includegraphics[width=\textwidth]{figs/simulations/exp1-thresh.png}
  \caption{FCFS for small jobs}
  \label{subfig:app-low-fcfs}
\end{subfigure}\hfill
\begin{subfigure}[b]{0.45\textwidth}
  \includegraphics[width=\textwidth]{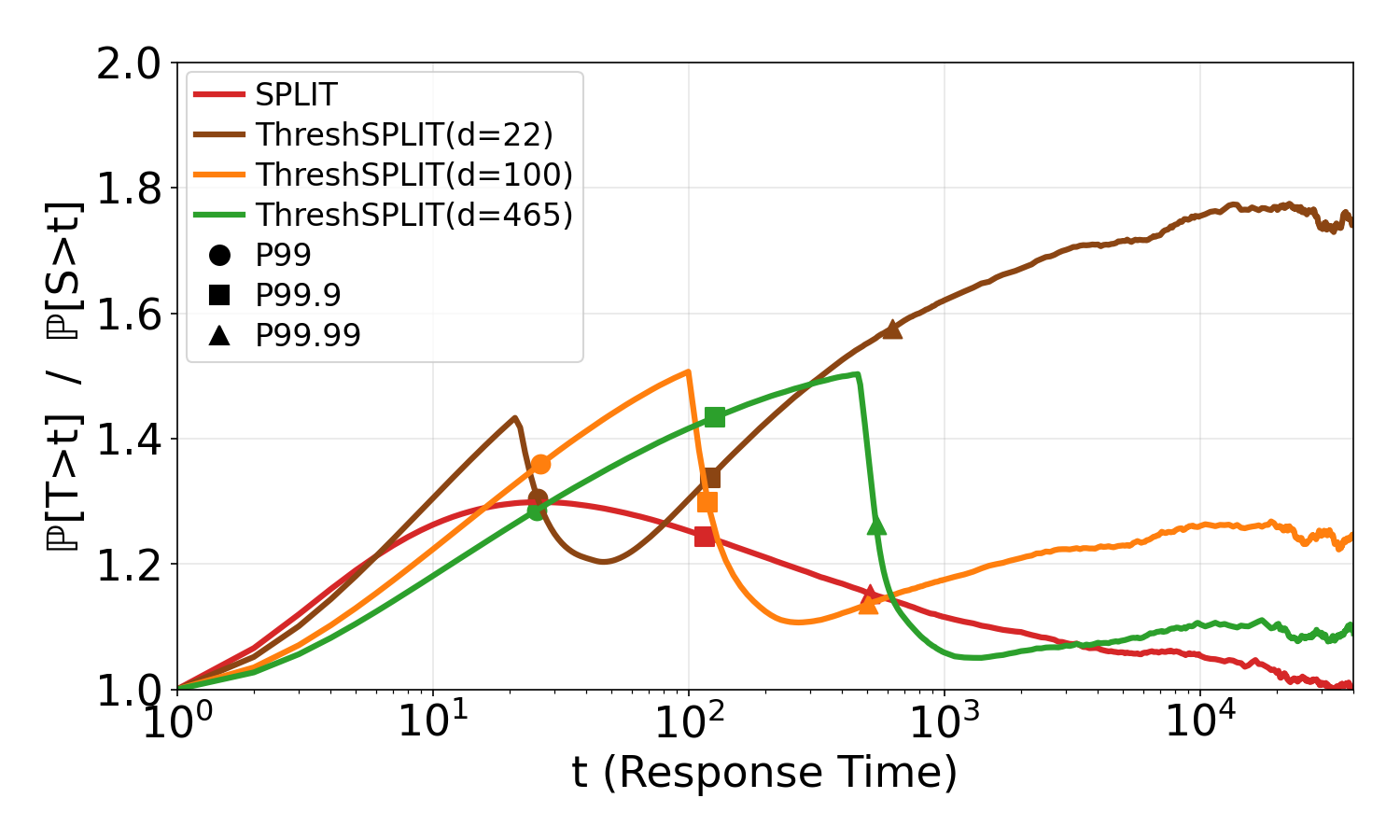}
  \caption{SRPT-$(n-1)$ for small jobs}
  \label{subfig:app-low-srpt}
\end{subfigure}
\caption{\threshsplit{d}, low-load setting ($n=3$, $\rho=0.5$, $\alpha=1.5$).}
\label{fig:app:fcfs-vs-srpt-low}
\end{figure}

\begin{figure}[h]
\centering
\begin{subfigure}[b]{0.45\textwidth}
  \includegraphics[width=\textwidth]{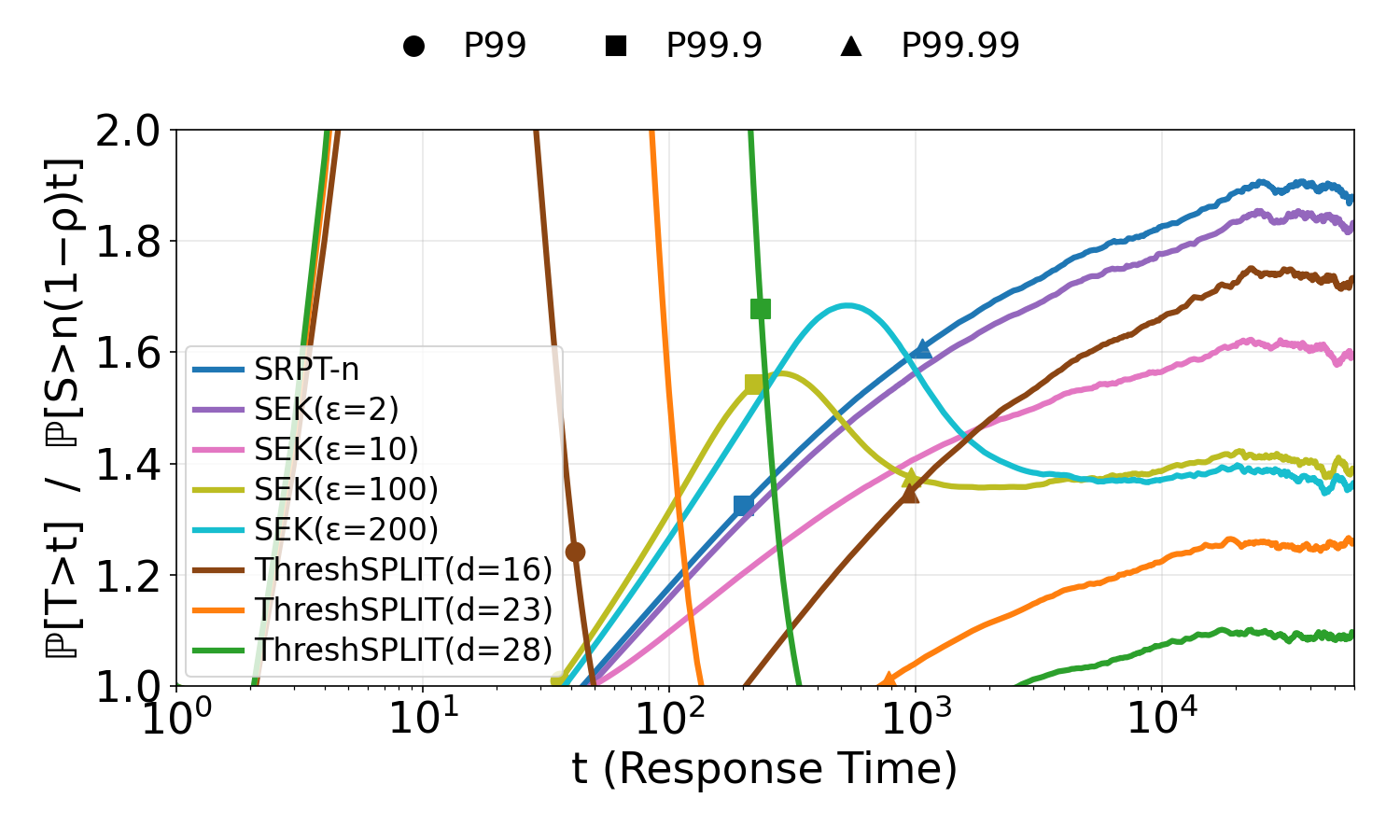}
  \caption{FCFS for small jobs}
  \label{subfig:app-exp2-fcfs}
\end{subfigure}\hfill
\begin{subfigure}[b]{0.45\textwidth}
  \includegraphics[width=\textwidth]{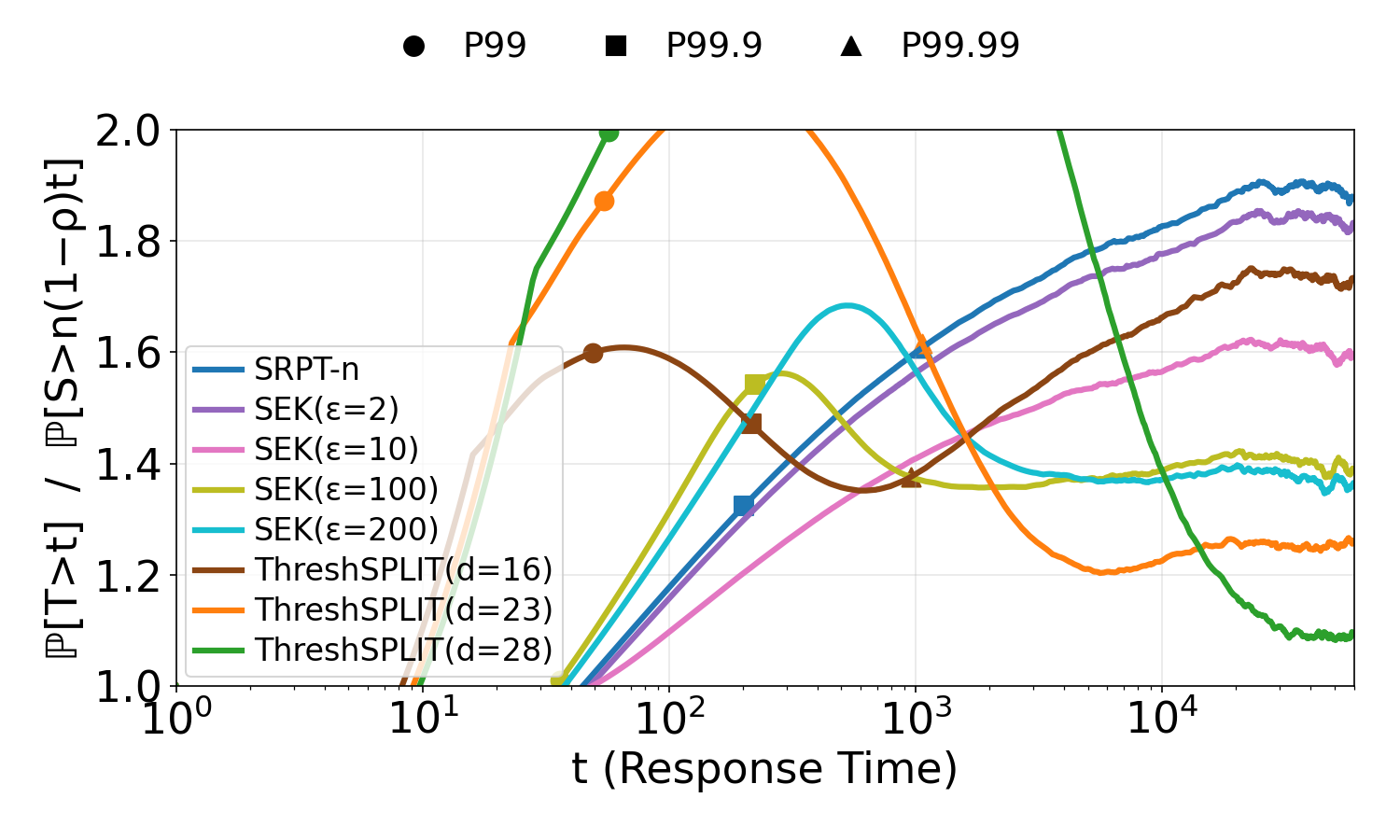}
  \caption{SRPT-$(n-1)$ for small jobs}
  \label{subfig:app-exp2-srpt}
\end{subfigure}
\caption{\threshsplit{d}, high-load setting ($n=3$, $\rho=0.8$, $\alpha=1.5$).}
\label{fig:app:fcfs-vs-srpt-exp2}
\end{figure}

\begin{figure}[h]
\centering
\begin{subfigure}[b]{0.45\textwidth}
  \includegraphics[width=\textwidth]{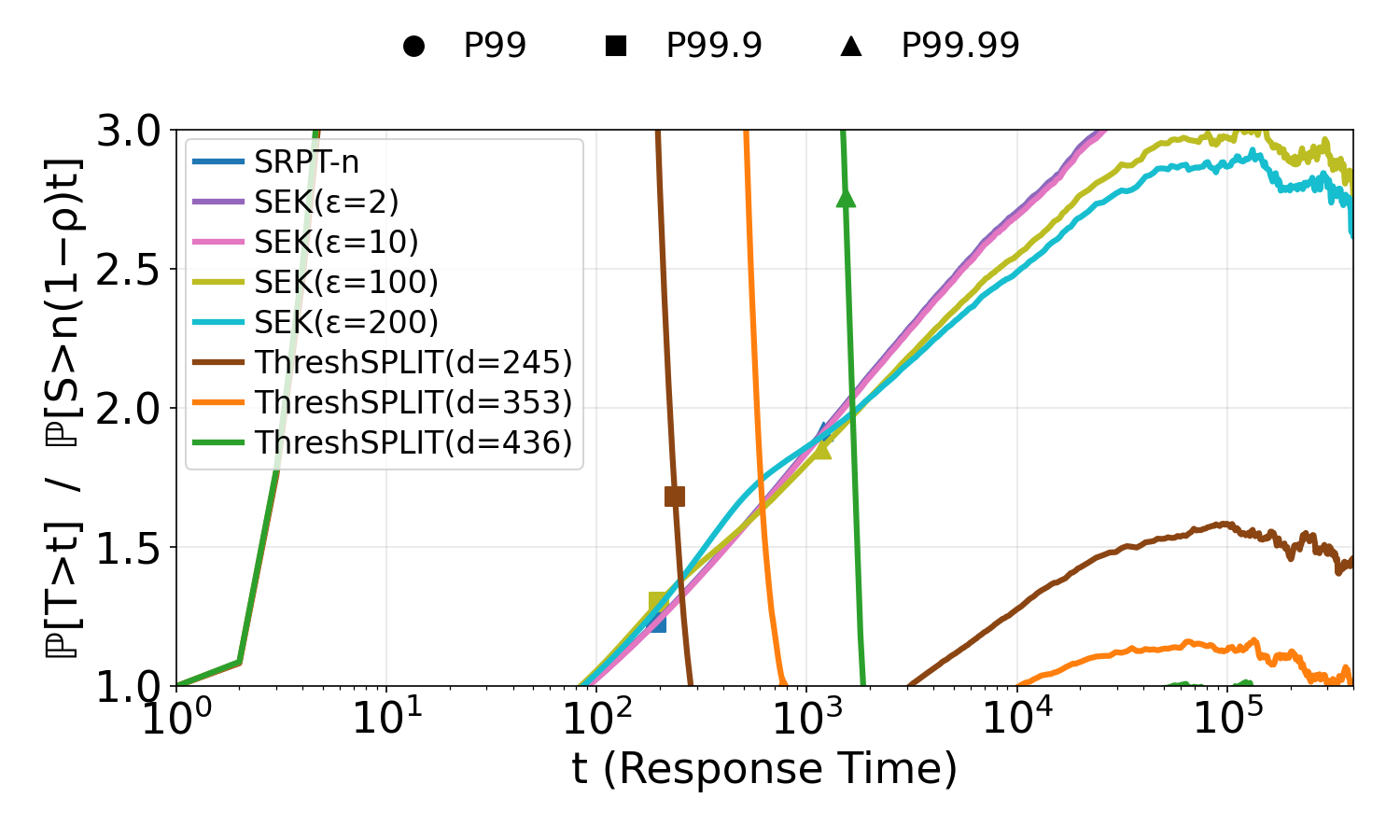}
  \caption{FCFS for small jobs}
  \label{subfig:app-exp3-fcfs}
\end{subfigure}\hfill
\begin{subfigure}[b]{0.45\textwidth}
  \includegraphics[width=\textwidth]{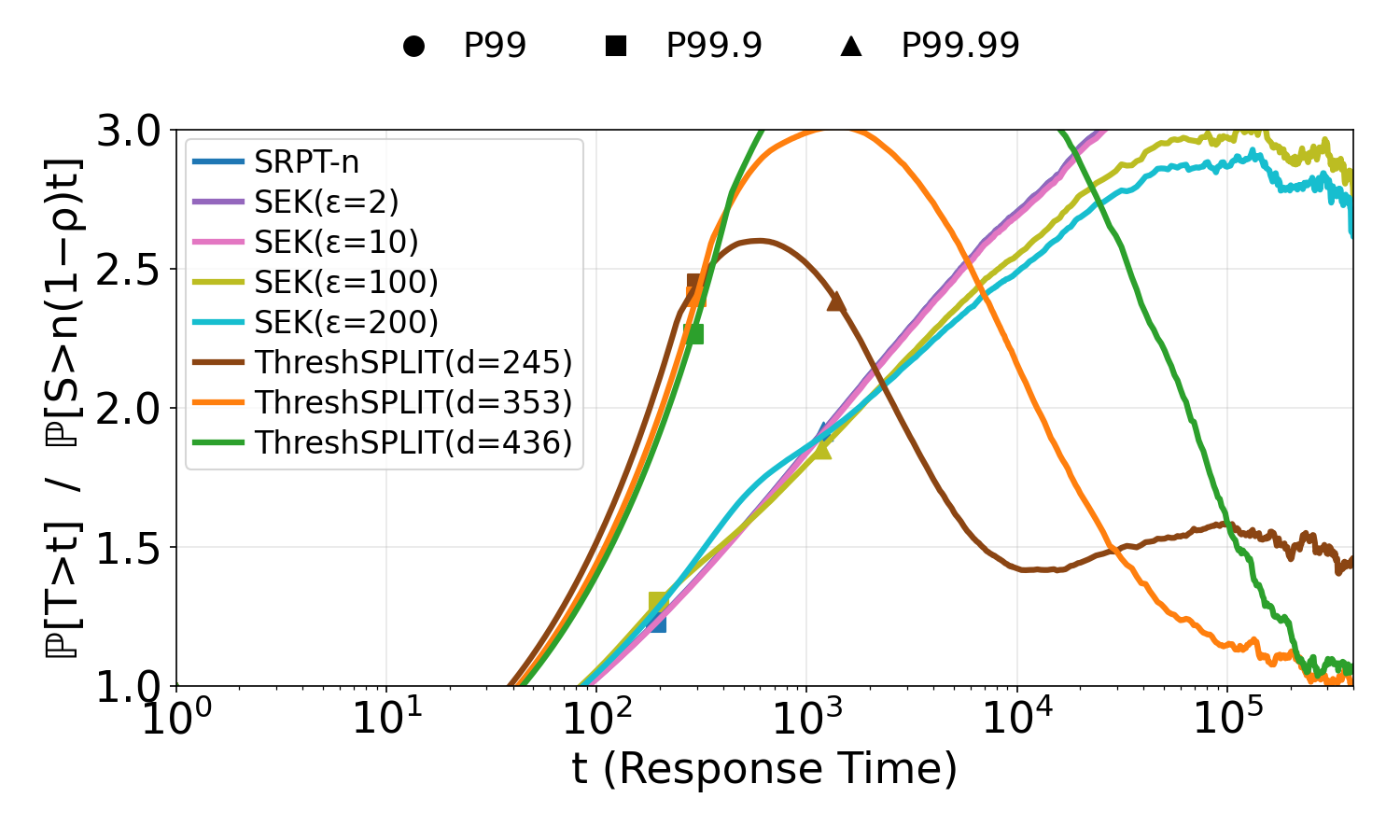}
  \caption{SRPT-$(n-1)$ for small jobs}
  \label{subfig:app-exp3-srpt}
\end{subfigure}
\caption{\threshsplit{d}, higher-load setting ($n=10$, $\rho=0.94$, $\alpha=1.5$).}
\label{fig:app:fcfs-vs-srpt-exp3}
\end{figure}

\subsection{Simulation figures with SEK policies for different parameters}
\label{sec:appendix:sek-epsilon}

In Section~\ref{sec:simulation}, the figures in the main text show only the SEK-$\epsilon$ curve with the largest $\epsilon$ in our sweep, since this gives the best asymptotic tail among the SEK family. For completeness, this appendix replicates Figures~\ref{subfig:lowload-main}, \ref{subfig:highload-exp2}, \ref{subfig:highload-exp3}, \ref{subfig:alpha1.4}, and \ref{subfig:alpha2} with the full set of SEK-$\epsilon$ curves for several values of $\epsilon$. Figure~\ref{fig:app:sek-lowload} shows the low-load setting, Figure~\ref{fig:app:sek-highload} the high-load settings, and Figure~\ref{fig:app:sek-alpha} the varying tail index $\alpha$. Across all settings, we observe that SEK-$\epsilon$ improves on SRPT-$n$ as $\epsilon$ grows, but the asymptotic improvement saturates.

\begin{figure}[h]
\centering
\includegraphics[width=0.45\textwidth]{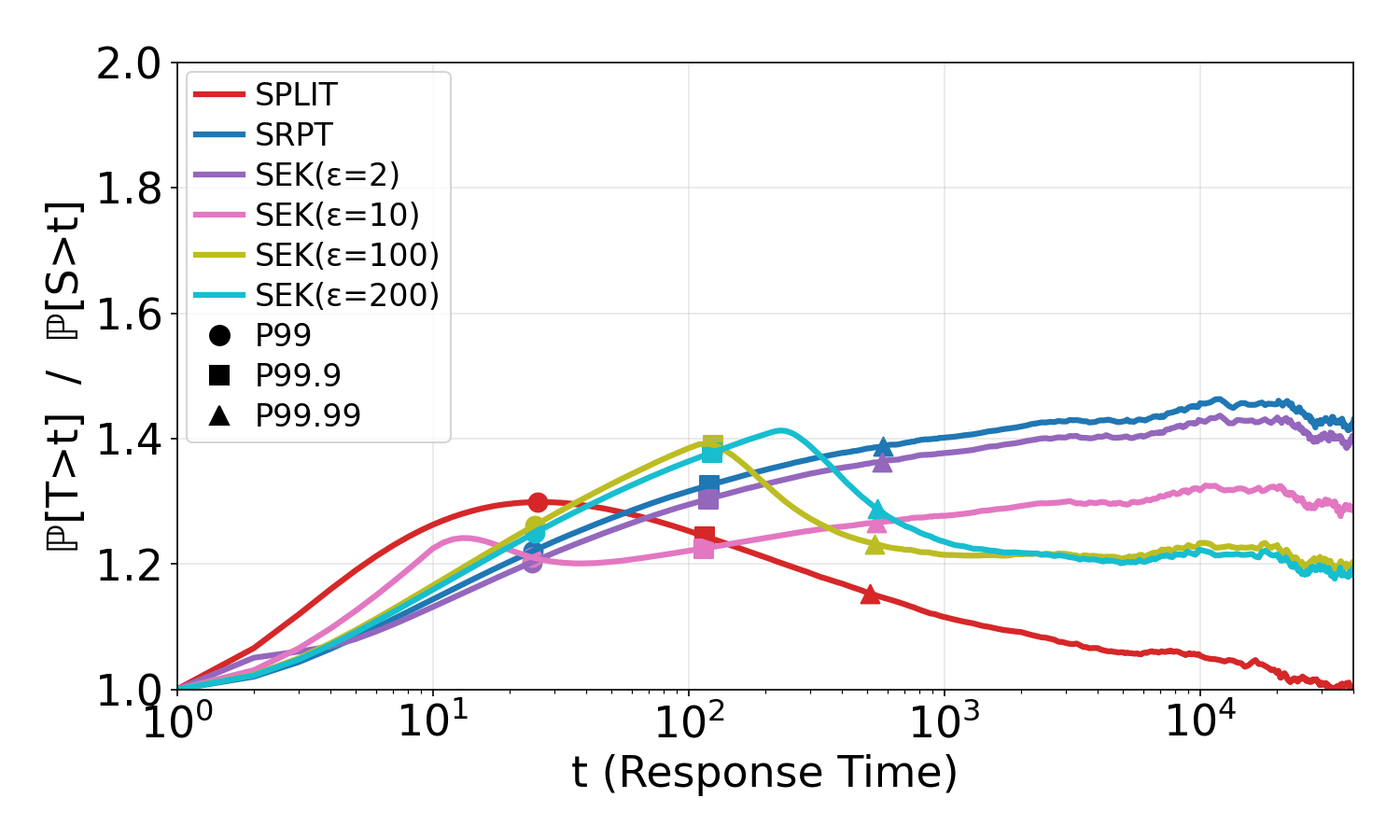}
\caption{Low-load setting ($n=3$, $\rho=0.5$, $\alpha=1.5$); cf.\ Figure~\ref{subfig:lowload-main}.}
\label{fig:app:sek-lowload}
\end{figure}

\begin{figure}[h]
\centering
\begin{subfigure}[b]{0.45\textwidth}
  \includegraphics[width=\textwidth]{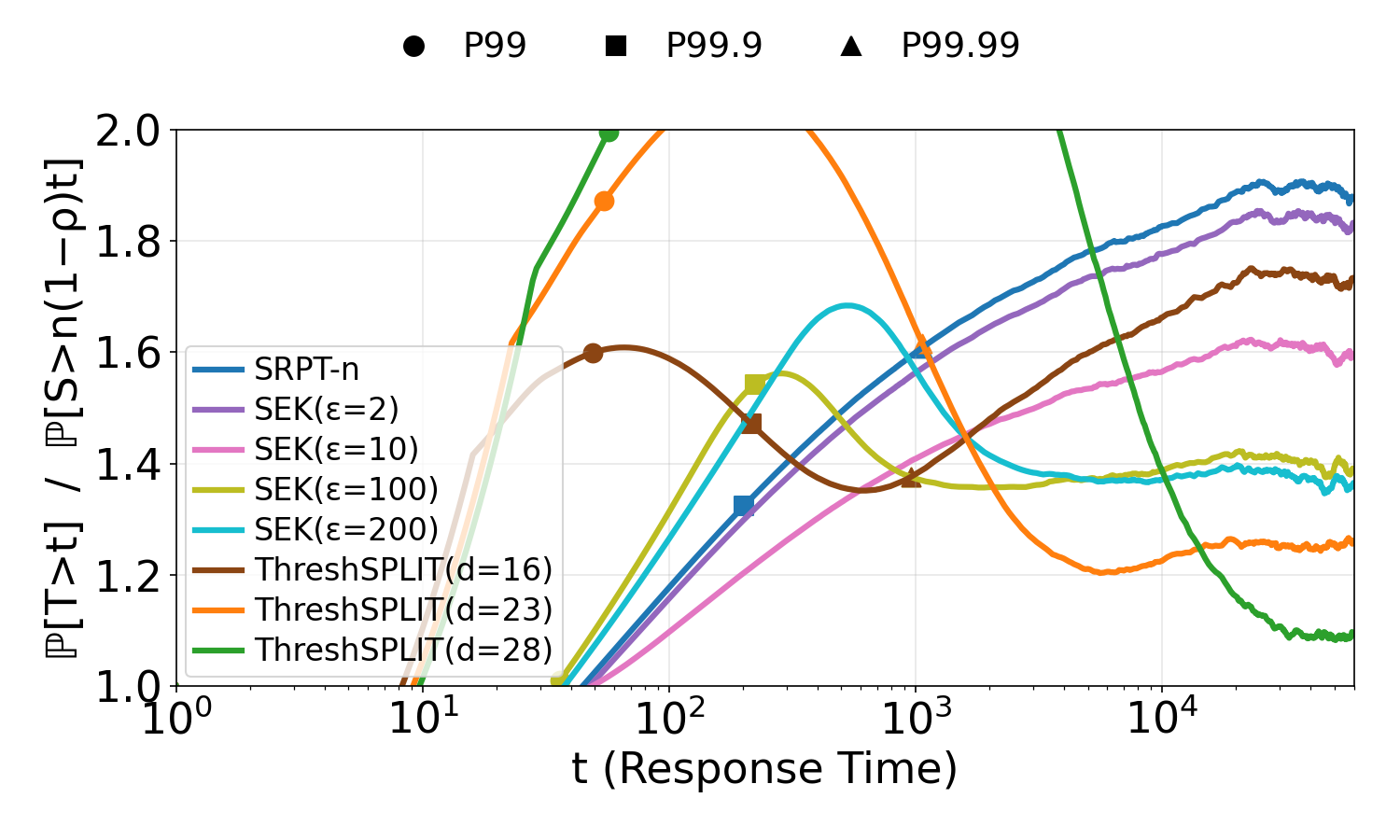}
  \caption{$n=3$, $\rho=0.8$; cf.\ Figure~\ref{subfig:highload-exp2}.}
  \label{subfig:app-sek-exp2}
\end{subfigure}\hfill
\begin{subfigure}[b]{0.45\textwidth}
  \includegraphics[width=\textwidth]{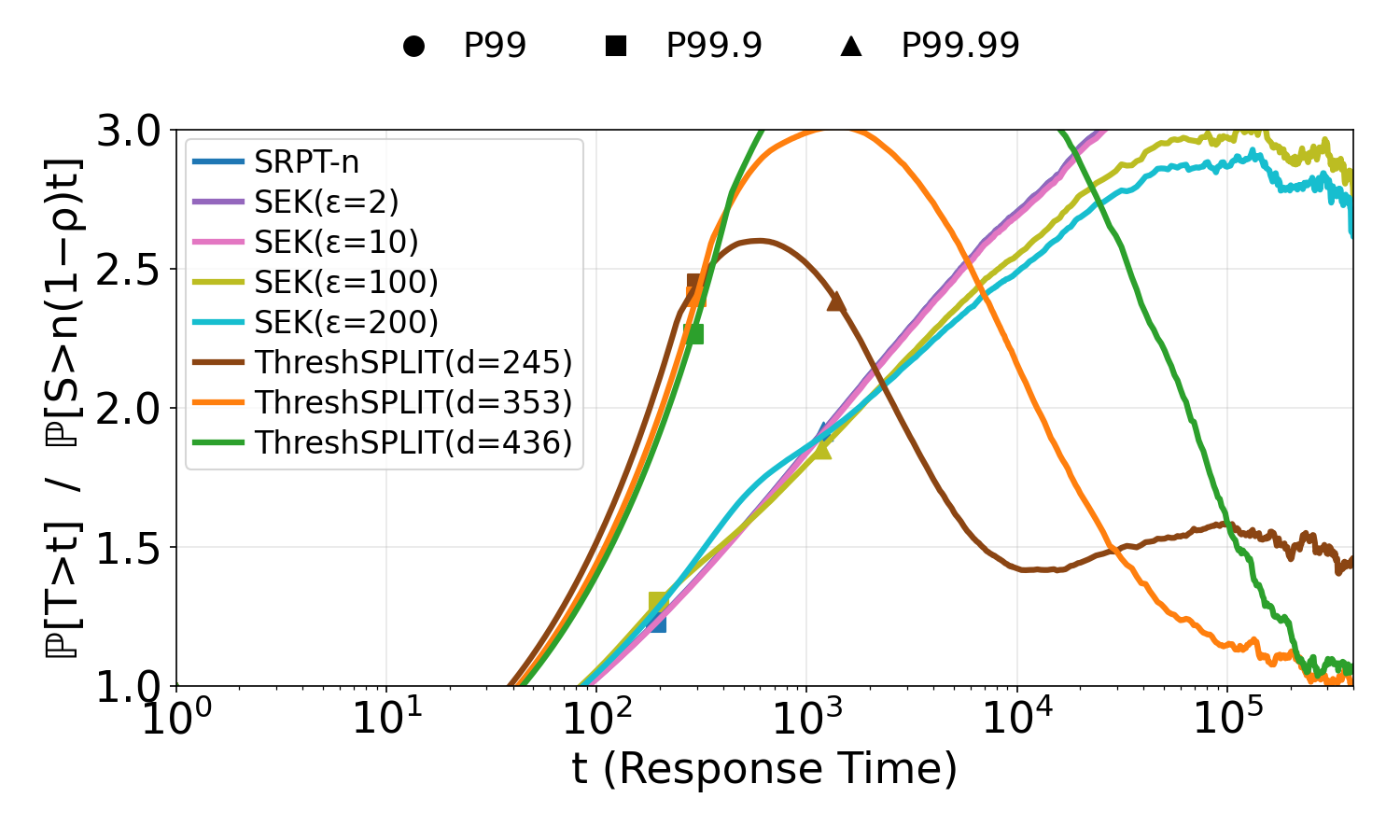}
  \caption{$n=10$, $\rho=0.94$; cf.\ Figure~\ref{subfig:highload-exp3}.}
  \label{subfig:app-sek-exp3}
\end{subfigure}
\caption{High-load settings, Pareto job sizes with $\alpha=1.5$.}
\label{fig:app:sek-highload}
\end{figure}

\begin{figure}[h]
\centering
\begin{subfigure}[b]{0.45\textwidth}
  \includegraphics[width=\textwidth]{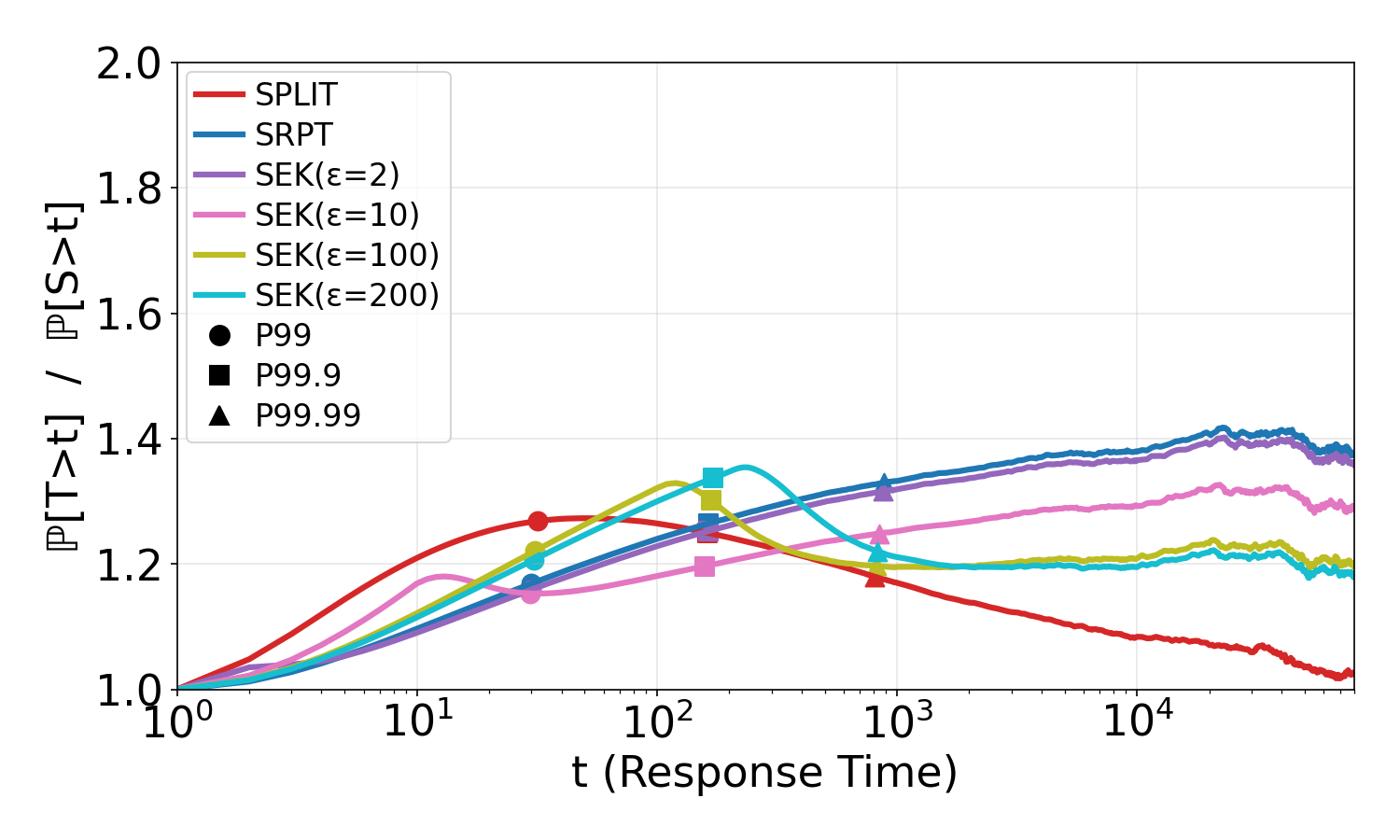}
  \caption{$\alpha=1.4$; cf.\ Figure~\ref{subfig:alpha1.4}.}
  \label{subfig:app-sek-alpha1.4}
\end{subfigure}\hfill
\begin{subfigure}[b]{0.45\textwidth}
  \includegraphics[width=\textwidth]{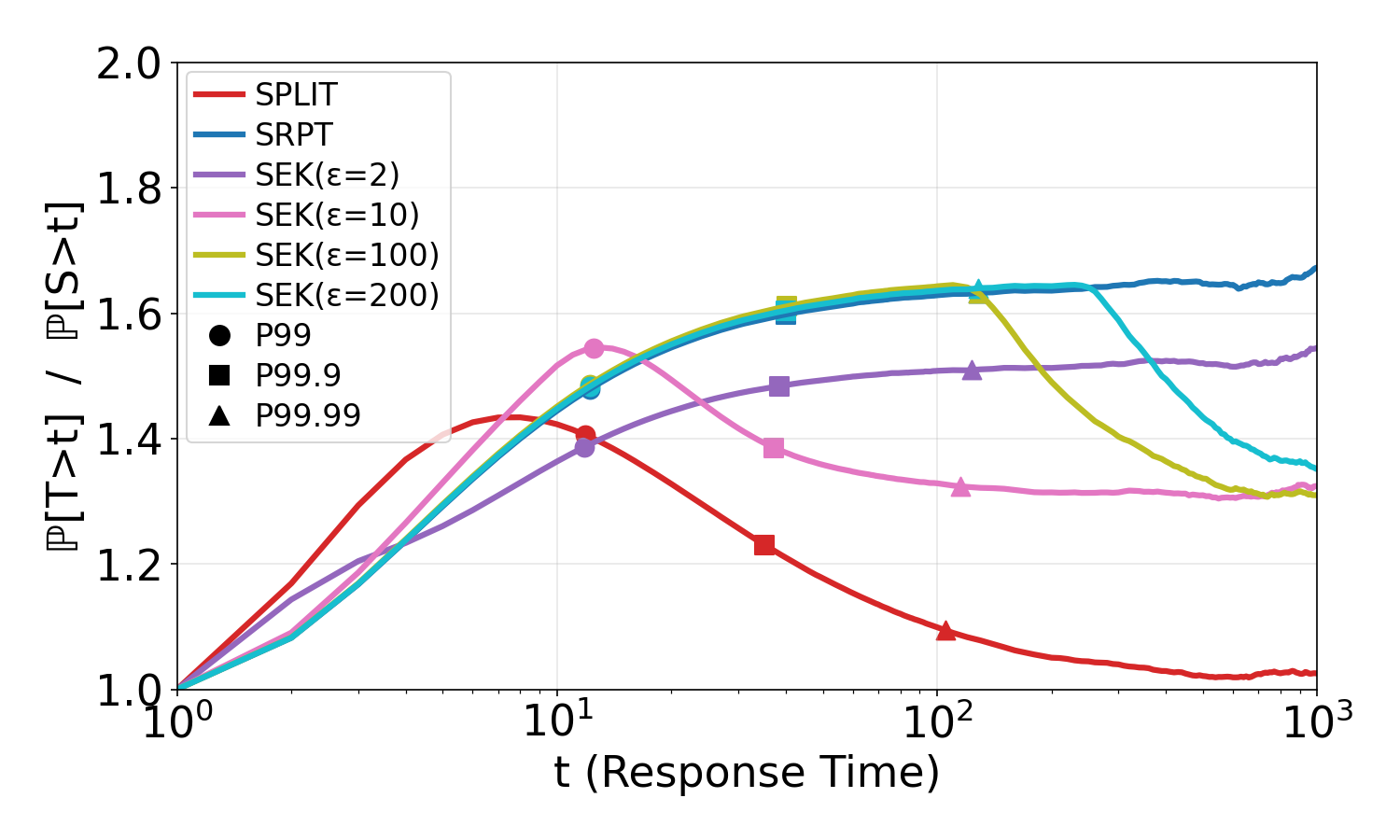}
  \caption{$\alpha=2.0$; cf.\ Figure~\ref{subfig:alpha2}.}
  \label{subfig:app-sek-alpha2}
\end{subfigure}
\caption{Varying Pareto tail index, $n=3$ servers, $\rho=0.5$.}
\label{fig:app:sek-alpha}
\end{figure}